\documentclass[11pt]{article}
\usepackage{amsmath}
\usepackage{multirow}
\usepackage{enumerate}
\usepackage[authoryear]{natbib}
\usepackage{url}
\usepackage[margin=1in]{geometry}

\usepackage{amsmath,amssymb,dsfont,color,bm,mathtools,enumitem}
\mathtoolsset{showonlyrefs}
\usepackage{amsthm}
\usepackage{subcaption}

\usepackage{algorithm, algpseudocode}

\algnewcommand\algorithmicinput{\textbf{Input:}}
\algnewcommand\algorithmicoutput{\textbf{Output:}}
\algnewcommand\INPUT{\item[\algorithmicinput]}
\algnewcommand\OUTPUT{\item[\algorithmicoutput]}

\theoremstyle{definition}
\newtheorem{thm}{Theorem}
\newtheorem{lem}{Lemma}
\newtheorem{prop}{Proposition}

\newtheorem{conj}{Conjecture}

\newtheorem{rmk}{Remark}
\newtheorem{example}{Example}

\usepackage[parfill]{parskip}
\usepackage{stmaryrd} 

\def\ma{\bm{a}}
\def\mb{\bm{b}}

\def\md{\bm{d}}

\def\mr{\bm{r}}
\def\ms{\bm{s}}

\def\bmu{\bm{u}}

\def\mx{\bm{x}}
\def\my{\bm{y}}
\def\mz{\bm{z}}

\def\mA{\bm{A}}
\def\mB{\bm{B}}
\def\mC{\bm{C}}

\def\mM{\bm{M}}

\def\mU{\bm{U}}
\def\mV{\bm{V}}

\def\mX{\bm{X}}

\def\mx{\bm x}

\def\mA{\bm A}
\def\mB{\bm B}

\def\mC{\bm C}

\def\mX{\bm X}

\def\mM{\bm M}
\def\mSigma{\bm \Sigma}

\def\tZ{\mathcal{Z}}
\def\tA{\mathcal{A}}
\def\tB{\mathcal{B}}
\def\tC{\mathcal{C}}

\def\tE{\mathcal{E}}

\def\tG{\mathcal{G}}
\def\tH{\mathcal{H}}

\def\tO{\mathcal{O}}
\def\tP{\mathcal{P}}

\def\tT{\mathcal{T}}

\def\tX{\mathcal{X}}
\def\tY{\mathcal{Y}}
\def\tZ{\mathcal{Z}}

\def\bbR{\mathbb{R}}

\newcommand{\SnormSize}[2]{#1\lVert#2#1\rVert_{\text{sp}}}
\newcommand{\FnormSize}[2]{#1\lVert#2#1\rVert_F}

\DeclareMathOperator*{\argmin}{arg\,min}

\usepackage{mathrsfs}

\usepackage{wrapfig}
\newcommand*{\KeepStyleUnderBrace}[1]{
  \mathop{%
    \mathchoice
    {\underbrace{\displaystyle#1}}%
    {\underbrace{\textstyle#1}}%
    {\underbrace{\scriptstyle#1}}%
    {\underbrace{\scriptscriptstyle#1}}%
  }\limits
}

\newcommand*{\vertbar}{\rule[-1ex]{0.5pt}{2.5ex}}
\newcommand*{\horzbar}{\rule[.5ex]{2.5ex}{0.5pt}}

\allowdisplaybreaks
\usepackage[colorlinks,citecolor=blue]{hyperref}

\title{Statistical and computational rates in high rank tensor estimation}
\date{}
\author{%
Chanwoo Lee \\
University of Wisconsin -- Madison\\
\texttt{chanwoo.lee@wisc.edu} \\
\and
Miaoyan Wang \\
University of Wisconsin -- Madison\\
\texttt{miaoyan.wang@wisc.edu} \\
}

\begin{document}
\maketitle

\begin{abstract}%
Higher-order tensor datasets arise commonly in recommendation systems, neuroimaging, and social networks. Here we develop probable methods for estimating a possibly high rank signal tensor from noisy observations. We consider a
generative latent variable tensor model that incorporates both high rank and low rank models, including but not limited to, simple hypergraphon models, single index models, low-rank CP models, and low-rank Tucker models. 
Comprehensive results are developed on both the statistical and computational limits for the signal tensor estimation. We find that high-dimensional latent variable tensors are of log-rank; the fact explains the pervasiveness of low-rank tensors in applications. Furthermore, we propose a polynomial-time spectral algorithm that achieves the computationally optimal rate. We show that the statistical-computational gap emerges only for latent variable tensors of order~3 or higher. Numerical experiments and two real data applications are presented to demonstrate the practical merits of our methods. 
  \end{abstract}

\noindent%

\section{Introduction}\label{sec:intro}
The analysis of higher-order tensors has recently drawn much attention in statistics, machine learning, and data science.
Higher-order tensor datasets are collected in applications including recommendation systems \citep{baltrunas2011incarmusic,bi2018multilayer}, social networks \citep{bickel2009nonparametric},
neuroimaging \citep{zhou2013tensor}, genomics \citep{hore2016tensor}, and longitudinal data analysis~\citep{hoff2015multilinear}. 
One example is a multi-tissue
expression data~\citep{wang2019three}. This dataset collects genome-wide expression profiles from different tissues in a number of individuals, which results in three-way tensor of gene $\times$ individual $\times$ tissue. Another example is hypergraph  networks, in which edges are allowed to connect more than two vertices. Considering multi-way interactions based on hypergraphs helps to understand complex networks in molecule system~\citep{michoel2012alignment} and computer vision~\citep{Agarwal2006HigherOL}. Tensors are naturally used to represent such hypergraph structures.
Along with many important applications, tensor methods have provided effectiveness in data analysis that classical vector- or matrix-based methods fail to offer~\citep{han2022exact,lee2021beyond}.

One of popular structures imposed on the tensor of interest is the low-rankness.  Common low rank models include CP low rank models~\citep{kolda2009tensor,sun2017provable}, Tucker low rank models~\citep{zhang2018tensor}, and block models~\citep{wang2019multiway}.
Despite the popularity of the low rank assumption, it is rather restricted to assume that the rank of the tensor remains fixed while the tensor dimension increases to infinity. In particular, low rank assumption is sensitive to entrywise transformation and inadequate for representing special structures of tensors~\citep{lee2021beyond}. In addition, low rank tensors are nowhere dense, and random matrices/tensors are almost surely of full rank~\citep{udell2019big}.  This motivates us to develop a more flexible model that can handle \emph{possibly high rank} tensors.

\subsection{Our contributions}

We develop \emph{a latent variable tensor model} that addresses both low and high rank tensors. Our model includes, but is not limited to, most existing tensor models such as CP models~\citep{kolda2009tensor}, Tucker models~\citep{zhang2018tensor}, generalized linear models~\citep{wang2018learning, hu2021generalized}, single index models~\citep{ganti2017learning}, and simple hypergraphon models~\citep{balasubramanian2021nonparametric}. Table~\ref{tab:comp} compares our work with previous results from both statistical and computational perspectives, which we summarize below.

First, we provide a rigorous justification for the empirical success of low-rank methods despite the prevalence of high rank tensors in real data applications. We prove that $d$-dimensional tensors generated from latent variable tensor models are of log-rank $\tO(\log d)$ as $d\to\infty$. This key spectral property provides the rational of low-rank approximation from a statistical perspective. 

\begin{table}
    \centering
    \caption{Comparison of our results with previous works on generalized linear models~\citep{hu2021generalized,wang2018learning,zhang2018tensor} and simple hypergraphon models~\citep{balasubramanian2021nonparametric}. We propose two methods, the least-square estimator (LSE) and the double-projection spectral estimator (DSE). Our results allow high rank tensor estimation and provide statistical and computational optimality on mean square error (MSE) rates. $^*$For generalized linear models, polynomial algorithms exist only under the strong signal-to-noise ratio. }
    \label{tab:comp}
    \resizebox{\textwidth}{!}{
    \begin{tabular}{c|@{\hskip4pt}c@{\hskip5pt}c@{\hskip4pt}c@{\hskip5pt}c@{\hskip4pt}c@{\hskip5pt}c@{\hskip4pt}c@{\hskip5pt}c}
    & Generalized linear models&  Simple hypergraphon models &  \textbf{Ours} (LSE)&\textbf{Ours} (DSE)\\
    \hline
     MSE rate for order-$m$ tensor & {$d^{-(m-1)}$} & $d^{-{2m/(m+2)}}$ & $d^{-(m-1)}$ & $d^{-m/2}$ \\
       (e.g., when $m=3$) &$d^{-2}$ & ($d^{-6/5}$) & $(d^{-2})$ & $(d^{-3/2})$ \\
        Allows high-rankness & $\times$  & $\surd$ & $\surd$ & $\surd$ \\
        Optimality analysis & $\surd$ & $\times$ & $\surd$ &$\surd$\\
     Polynomial algorithm& $\times/\surd^*$ &$\times$ & $\times$ & $\surd$\\
    \end{tabular}
    }
\end{table}

Second, we discover the gap between statistical and computational optimality in the higher-order tensor estimation. 
We show that the \emph{statistically} minimax optimal rate of the problem is $d^{-(m-1)}$. We prove that this rate, however, is non-achievable by any polynomial-time algorithms under hypergraphic planted clique (HPC) conjecture~\citep{luo2022tensor}. We then show that a slower rate $d^{-m/2}$ is \emph{computationally} optimal and achievable by polynomial-time algorithms. Based on these two bounds, we reveal the gap regime $[d^{-(m-1)},\ d^{-m/2}]$ where the estimation is statistically possible but computationally impossible. This phenomenon is distinctive from matrix problems with $m=2$ where no gap regime exists. Figure~\ref{fig:semantic} illustrates this statistical-computational gap in the higher-order tensor estimation. 


\begin{figure}
    \centerline{\includegraphics[width =\textwidth]{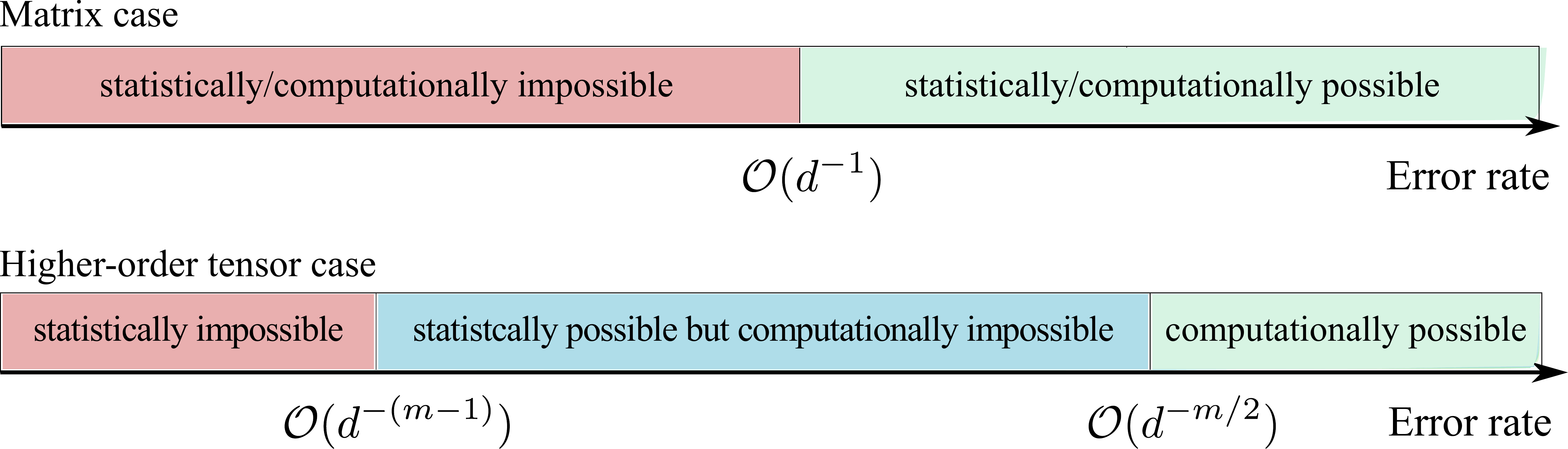}}
    \caption{Statistical and computational boundaries of the mean square error in the high-rank tensor estimation problem. Here, the signal is an order-$m$, $(d,\ldots,d)$-dimensional tensor. We find that the statistical-computational gap arises only for higher-order tensors with $m\geq 3$.}
    \label{fig:semantic}
\end{figure}

Third, we propose two estimation methods with accuracy guarantees: the least-square estimation (LSE) and double-projection spectral estimation (DSE).  The LES achieves the information-theoretical lower bound $d^{-(m-1)}$ demonstrating its statistical optimality. The computation of LSE, however, requires possibly non-polynomial complexity. We then propose the DSE using the idea of double-projection spectral method~\citep{zhang2018tensor} and the log-rank property of latent variable tensors. We show that the DSE achieves the optimal bound $d^{-m/2}$ within the subclass of polynomial-time estimators. 

Finally, we illustrate the efficacy of our methods through simulation and data examples. We apply our model to fMRI brain image and crop production datasets. Our estimator successfully recovers the true brain image from noisy observations. Clustering analysis on the crop production dataset reveals regional patterns of countries. Numerical analysis demonstrates the practical utility of the proposed approach.

\subsection{Related work}
This work is related to but also clearly distinct from a broad range of literature on tensor analysis. We review several lines of related research for comparison.

\paragraph*{Low-rank based tensor models} 
The past decades have seen a large body of work on structured tensor estimation under low rank models, including CANDECOMP/PARAFAC (CP) models~\citep{kolda2009tensor,sun2017provable}, Tucker models~\citep{zhang2018tensor}, and block models~\citep{wang2019multiway,han2022exact}. Single index models~\citep{ganti2017learning} and generalized linear models~\citep{wang2018learning, hu2021generalized, han2022optimal} have been proposed to overcome the low rank assumption on the signal tensor. These works, however, still assume the low-rankness on the underlying latent tensor up to link functions. Low-rank based models belong to parametric approaches because they model the data tensor using a finite number of parameters, i.e., a set of $r$ decomposed factors. In contrast, our method does not assume the low rank structure. We consider the high-rankness from latent variable tensor models and take a nonparametric approach by allowing an infinite number of parameters. The benefits of nonparametric methods over parametric ones have been reported in many statistical problems~\citep{pananjady2022isotonic,gao2015rate, bickel2009nonparametric}.

\paragraph*{Graphon and hypergraphon}   Research on graphon and hypergraphon is connected to our work. The graphon is a measurable function representing the limit of a sequence of exchangeable random graphs~\citep{udell2019big,xu2018rates,klopp2017oracle,gao2015rate,chan2014consistent}. Similar to graphons, the hypergraphon~\citep{zhao2015hypergraph,lovasz2012large} is a limiting function of $m$-uniform hypergraphs whose edges can join $m$ vertices with $m\geq 3$. The hypergraphon provides a powerful tool for modeling multi-way interactions, as the graphon does for pair-wise interactions. Unlike the matrices where bi-variate functions are enough to represent graphons~\citep{lovasz2006limits}, general hypegraphons are represented as $(2^m-2)$-multivariate functions~\citep{zhao2015hypergraph}. The simple hypergraphon~\citep{balasubramanian2021nonparametric} considers $m$-multivariate functions as a tradeoff between model flexibility and efficient estimation. The simple hypergraphon shares the common ground with our latent variable tensor model in the sense that both consider $m$-multivariate functions. However, our models use vector-valued latent variables while simple hypergraphon uses scalar-valued latent variables. The comparison of two approaches will be discussed in Section~\ref{sec:disc}.

\paragraph*{High rank tensor estimation} A few recent attempts have been made to analyze high rank signal tensors. For example, the work in \cite{lee2021beyond} proposes a new notion of tensor rank called sign-rank that can model high rank tensors with low complexity. While this work provides a polynomial-time algorithm, the proposed algorithm needs to solve sub-optimization problems as many as the order of the tensor dimension and the minimax optimality is unknown. 
The work in \cite{lee2021smooth} proposes another high rank model called permuted smooth tensor models. The permuted smooth tensor model is similar to simple hypergraphons~\citep{balasubramanian2021nonparametric}, in that both assume the latent scalar-valued latent variables. In contrast, our model allows general vector-valued latent variables. 
In algorithmic perspective, \cite{lee2021smooth} requires strict monotonic assumption on the latent function to achieve algorithm accuracy. 
Instead, our double-projection spectral algorithm requires no any further assumptions on the model and achieves computationally optimal rate. Detailed comparisons will be performed in Section~\ref{sec:sim}.

\paragraph*{The statistical-computational gap in tensor problems} Another related topic is on the statistical-computational gap in tensor problems. The existence of intrinsic statistical-computational gap have been found in tensor completion~\citep{barak2016noisy,wang2018learning},  tensor PCA~\citep{richard2014statistical,zhang2018tensor,han2022optimal}, and multiway clustering~\citep{luo2022tensor}. Due to its importance of understanding the fundamental difference between matrices and tensors, the computational limit analysis has gained enormous attention. 
For the matrices, the computational limits are derived from average-case reduction scheme based on the common hardness assumption of the planted clique problem in the Erd\"os-R\'enyi graph~\citep{hazan2011hard,ma2015computational,berthet2020statistical,wu2021statistical}. However, direct application of planted clique detection to tensors is complicated by the multi-way structure of the tensor~\citep{luo2020open}.
More recently, planted hypergraphic clique has been proposed to overcome this problem~\citep{zhang2018tensor,luo2022tensor}. 
Our work is also in line with understanding the statistical-computational gap, but for a more challenging high rank tensor estimation problem.

\subsection{Notation and organization}
We use $\mathbb{R}$, $\mathbb{N}$, $\mathbb{N}_{+}$ to denote the set of real numbers, integers, and positive integers, respectively. Let $[d]=\{1,\ldots,d\}$ denote the $d$-set with $d\in\mathbb{N}_{+}$. For two positive sequences $\{a_d\},\{b_d\}$,  we denote $a_d\lesssim b_d$ if $\lim_{d\to\infty} a_d/b_d\leq c$ for some constant $c>0$, and $a_d\asymp b_d$ if $c_1\leq \lim_{d\to \infty} a_d/b_d\leq c_2$ for some constants $c_1,c_2>0$. We use $\tO(\cdot)$ to denote big-O notation, and $\tilde\tO(\cdot)$ the variant hiding logarithm factors. We use bold lowercase letters (e.g., $\ma, \mb, \bmu$) for vectors, and bold uppercase letters (e.g., $\mA, \mB, \mC$) for matrices. For a matrix $\mA\in\mathbb{R}^{d_1\times d_2}$, $\textup{SVD}_r(\mA)$ denotes matrix comprised of the top $r$ left singular vectors of $\mA$, and $\SnormSize{}{\mA}$ denotes the spectral norm of $\mA$. 

Let $\Theta\in\bbR^{d_1\times \cdots \times d_m}$ be an order-$m$ $(d_1,\ldots,d_m)$-dimensional tensor. We use $\Theta(i_1,\ldots,i_m)$ to denote the tensor entry indexed by $(i_1,\ldots,i_m)\in[d_1]\times\cdots\times[d_m]$. We define the Frobenius norm and the $\infty$-norm of a tensor $\Theta$ as 
\[
\FnormSize{}{\Theta} = \sqrt{\sum_{(i_1,\ldots,i_m)\in[d_1]\times \cdots\times [d_m]}\Theta^2(i_1,\ldots,i_m)},\ \ \|\Theta\|_\infty = \max_{(i_1,\ldots,i_m)\in[d_1]\times \cdots\times [d_m]}|\Theta(i_1,\ldots,i_m)|.
\]
The multilinear multiplication of a tensor $\tC\in\mathbb{R}^{r_1\times \cdots\times r_m}$ by matrices $\mU^{(k)}\in\mathbb{R}^{d_k\times r_k}$, $k\in[m]$ is defined as 
\begin{align}
    (\tC\times_1\mU^{(1)}\times\cdots&\times_m\mU^{(m)})(i_1,\ldots,i_m) \\
   & \quad \quad  \quad \quad  = \sum_{j_1 = 1}^{r_1}\cdots\sum_{j_m = 1}^{r_m}\tC(j_1,\ldots,j_d)\mU^{(1)}(i_1,j_1)\cdots\mU^{(m)}(i_m,j_m),
\end{align}
which results in an order-$m$ $(d_1,\ldots,d_m)$-dimensional tensor. We use $\text{Unfold}_k(\cdot)$ to denote the unfolding operation that reshapes the tensor along mode $k$ into a matrix, for $k\in[m]$. We say a tensor $\Theta\in\mathbb{R}^{d_1\times \cdots\times d_m}$ has Tucker-rank $(r_1,\ldots,r_m)$ if $r_k = \textup{rank}\left(\textup{Unfold}_k(\Theta)\right)$.  
Throughout this paper, we reserve the term ``tensor rank'' for the Tucker-rank defined above, unless stated otherwise. An event $E$ is said to occur \emph{with high probability} if $\mathbb{P}(E)$ tends to 1 as the tensor dimension $\underline d\to\infty$. Finally, for an $m$-dimensional vector $\md = (d_1,\ldots,d_m)$, we use the following shorthand notation,
\begin{align}
    \bar d := \max_{k\in[m]}d_k, \quad \underline d := \min_{k\in[m]}d_k, \quad d_* := \prod_{k\in[m]}d_k.
\end{align}
For an $s$-dimensional vector $\mx=(x_1,\ldots,x_m)$ and an integer partition $\alpha=\alpha_1+\cdots+\alpha_s$ with $\alpha_i\in\mathbb{N}$, we use the shorthand notation 
\begin{equation}\label{eq:short}
\partial \mx^\alpha= \partial x_1^{\alpha_1}\cdots \partial x_s^{\alpha_s}.
\end{equation}

The rest of the paper is organized as follows. In, Section~\ref{sec:md}, we propose a latent variable tensor model. The model allows high rank tensors and incorporates existing tensor models in past literature as special cases. 
We find that nice high-dimensional latent variable tensors are of log-rank; the fact explains the pervasiveness of low-rank tensors in applications. 
In Section~\ref{sec:lse}, we establish the minimax rate of the problem and the corresponding least-square estimator.  Section~\ref{sec:comp} presents that no polynomial-time algorithms can achieve the statistical optimality, thereby revealing the gap between statistical and computationally optimal performance. Then, we provide a polynomial-time double-projection spectral algorithm that achieves computationally optimal rate. Synthetic and real data analyses are presented in Section~\ref{sec:sim}.  The proofs for the
main theorems are provided  in Section~\ref{sec:proof}.  We conclude the paper with a  discussion in Section~\ref{sec:disc}. All technical lemmas and additional results are deferred to Appendix.

\section{A latent variable model for higher-order tensors}\label{sec:md}
\subsection{Model formulation}\label{sec:lvm}
Suppose we observe an order-$m$ $(d_1,\ldots,d_m)$-dimensional data tensor $\tY$ generated from the model 
\begin{align}\label{eq:gmodel}
    \tY = \Theta + \tE,
\end{align}
where $\Theta\in\mathbb{R}^{d_1\times \cdots \times d_m}$ is an unknown signal tensor of interest, and $\tE$ is a noise tensor consisting of independently and identically distributed (i.i.d.) standard Gaussian random variables. 

\paragraph*{Latent variable tensor model} The signal tensor $\Theta$ is generated by the following model.
\begin{itemize}[itemsep = 5pt]
    \item Nonparametric function: there exists a latent (unknown) multivariate function $f\colon\mathbb{R}^{s_1}\times\cdots\times\mathbb{R}^{s_m}\rightarrow \mathbb{R}$ such that 
    \begin{align}\label{eq:LVM}
    \Theta(i_1,\ldots,i_m) = f(\ma^{(1)}_{i_1},\ldots,\ma^{(m)}_{i_m}), \text{ for all } (i_1,\ldots,i_m)\in[d_1]\times\cdots\times [d_m],
\end{align}
where $\ma_{i_k}^{(k)}\in\mathbb{R}^{s_k}$ denotes the latent $s_k$-dimensional vector for each $k\in[m]$ and $i_k\in[d_k]$.
\item Regularity assumptions:\vspace{0.2cm}
\begin{itemize}
    \item Complete latent space: the latent vectors are supported on closed balls, i.e., $\ma_{i_k}^{(k)}\in B(\mathbb{R}^{s_k},\|\cdot\|_\infty)$ for all $k\in[m]$ and $i_k\in[d_k]$.  Here $B(\mathbb{R}^{s_k},\|\cdot\|_\infty)$ denotes an unit ball in $\mathbb{R}^{s_k}$ with respect to the infinity norm.
   
    \item Analytic latent function: we assume that the latent function $f$ is $M$-analytic such that
    \begin{align}\label{eq:analytic}
    \sup_{\mx_k\in B(\mathbb{R}^{s_k},\|\cdot\|_\infty),k\in[m]}\left|{\partial^{|\boldsymbol{\alpha}|} f(\mx_1,\ldots,\mx_m)\over \partial (\mx_1,\ldots,\mx_m)^{\bm \alpha}}\right|\leq M^{|\boldsymbol{\alpha}|}\boldsymbol{\alpha}!,
    \end{align}
\end{itemize}
for some constant $M\in\mathbb{R}$ and all multi-indices $\boldsymbol{\alpha}=(\alpha_1,\ldots,\alpha_m)\in\mathbb{N}^{m}$. Here, we have adopted the shorthand notation $|\boldsymbol{\alpha}|:=\sum_{k\in[m]}\alpha_k$, $\boldsymbol{\alpha}!:=\prod_{k\in[m]}\alpha_k!$, and~\eqref{eq:short}.  
\vspace{0.2cm}
\end{itemize}

We use $\tP(\md,\ms,M)$ to denote the latent variable tensor model under above assumptions, where $\md= (d_1,\ldots,d_m)$ represents the tensor dimension in~\eqref{eq:gmodel}, $\ms = (s_1,\ldots,s_m)$ represents the latent dimension in~\eqref{eq:LVM}, and $M\in\mathbb{R}$ is the regularity constant  in~\eqref{eq:analytic}. We write $\Theta\in\tP(\md,\ms,M)$ for the signal tensor in~\eqref{eq:gmodel}.

%

\vspace{.2cm}
\noindent
We next show that our latent variable tensor model is a broad family that incorporates many existing tensor models as special cases. 
\begin{example}[CP low rank models]\label{ex:cp}
The CP low rank tensors are the one of the 
popular tensor models~\citep{hitchcock1927expression,kolda2009tensor}. We now show that the CP low rank tensors belong to our latent variable tensor model. Let $\Theta$ be a low rank tensor with CP $s$-rank such that
\begin{align}
    \Theta = \sum_{r = 1}^s \lambda_r\bmu_r^{(1)}\otimes\cdots\otimes \bmu_r^{(m)},
\end{align}
where $\lambda_r>0$ and $\bmu_r^{(k)}\in B(\mathbb{R}^{d_k},\|\cdot\|_\infty)$ for all $r\in[s]$ and $k\in[m].$ Here $\otimes$ denotes the outer product of vectors. Define a vector $\bm \lambda= (\lambda_1,\ldots,\lambda_s) \in\mathbb{R}^s$ and factor matrices $\mA^{(k)} \in\mathbb{R}^{d_s\times s}$ as 
\[
\mA^{(k)} = \begin{pmatrix}
\vertbar & \vertbar & &\vertbar\\
 \bmu_1^{(k)}&\bmu_2^{(k)}&\cdots&\bmu_s^{(k)}\\
\vertbar &\vertbar & &\vertbar\\
 \end{pmatrix}
 =
 \begin{pmatrix}
\horzbar & (\ma_1^{(k)})^T &\horzbar\\
\horzbar & (\ma_2^{(k)})^T &\horzbar\\
&\vdots& \\
\horzbar & (\ma_{d_k}^{(k)} )^T &\horzbar\\
     \end{pmatrix},\quad \text{for all }k\in[m].
\]
Then, the CP $s$-rank tensor $\Theta$ belongs to our latent variable tensor model with $(s,\ldots,s)$-latent dimension satisfying
\begin{align}
    \Theta(i_1,\ldots,i_m) = f(\ma^{(1)}_{i_1},\ldots,\ma^{(m)}_{i_m}) \quad \text{for all } (i_1,\ldots,i_m)\in[d_1]\times\cdots\times [d_m],
\end{align}
where $\ma^{(k)}_{i_k}\in\mathbb{R}^{s}$ is $i_k$-th row of the factor matrix $\mA^{(k)}$, and $f\colon\mathbb{R}^{s}\times\cdots\times\mathbb{R}^s\rightarrow\mathbb{R}$ is an analytic function defined as $f(\mx_1,\ldots,\mx_m) = \left\langle \bm \lambda,\mx_1\circ\cdots\circ\mx_m\right\rangle$. Here $\circ$ is the element-wise product also known as the Hadamard product.
\end{example}

\begin{example}[Tucker low rank models]\label{ex:tucker}
The Tucker low-rankness is popularly imposed in tensor analysis~\citep{tucker1966some,kolda2009tensor}. We show that the Tucker low rank tensors also belong to our latent variable tensor model. Let $\Theta$ be a low rank tensor with Tucker $(s_1,\ldots,s_m)$-rank such that
\begin{align}
    \Theta = \tC\times_1\mA^{(1)}\times_2\cdots\times_m\mA^{(m)},
\end{align}
where $\tC\in\mathbb{R}^{s_1\times \cdots\times s_m}$ is a core tensor, $\mA^{(k)}\in\mathbb{R}^{d_k\times s_k}$ are factor matrices for $k\in[m]$.
The tensor $\Theta$ can be expressed as the latent variable tensor model with $(s_1,\ldots,s_m)$-latent dimension,
\begin{align}
    \Theta(i_1,\ldots,i_m) &= \tC\times_1 \ma_{i_1}^{(1)}\times_2\cdots\times \ma_{i_m}^{(m)}\\&=f(\ma_{i_1}^{(1)},\ldots,\ma_{i_m}^{(m)})\text{ for all } (i_1,\ldots,i_m)\in[d_1]\times_2\cdots\times[d_m],
\end{align}
where $\ma_{i_k}^{(k)}\in\mathbb{R}^{s_k}$ is the $i_k$-th row of the factor matrix $\mA^{(k)}$, and $f\colon\mathbb{R}^{s_1}\times \cdots\times \mathbb{R}^{s_m}\rightarrow\mathbb{R}$ is an analytic function defined as $f(\mx_1,\ldots,\mx_m) = \tC\times_1\mx_1^T\times \cdots\times_m\mx_m^T$.
\end{example}

\begin{example}[Generalized linear models]\label{ex:glm}
Let $\tY$ be a binary tensor from the logistic model~\citep{wang2018learning} with mean $\Theta(i_1,\ldots,i_m) = f(\tZ(i_1,\ldots,i_m))$, where $\tZ$ is a latent low rank tensor (with respect to CP or Tucker rank), and $f$ is the logistic link function defined as $f(x) = 1/(1+e^{-x})$. Notice that the signal tensor $\Theta$ itself is often \emph{high rank}.  Since the composition of analytic functions is analytic, we conclude the logistic model is included in our latent variable tensor model based on Examples~\ref{ex:cp} and \ref{ex:tucker}. Same conclusion holds for
general exponential-family models with an (known) analytic link function~\citep{hong2020generalized}.
\end{example}

\begin{example}[Single index models with an analytic function]
The single index model is a flexible semiparametric model
proposed in economics~\citep{robinson1988root} and high-dimensional statistics~\citep{balabdaoui2019least,ganti2017learning}. The single index models assume the existence of an \emph{unknown} link function $f\colon\mathbb{R}\rightarrow\mathbb{R}$ such that the signal tensor $\Theta(i_1,\ldots,i_m) = f(\tZ(i_1,\ldots,i_m))$, where $\tZ$ is a latent low rank tensor. Suppose the link function is $M$-analytic. We see that the single index model belongs to the latent variable tensor model by the same reasons in Example~\ref{ex:glm}.
\end{example}

\begin{example}[Simple hypergraphon models]\label{ex:graphon}
The graphon is a measurable function representing a limit of a sequence of exchangeable random
graphs (matrices)~\citep{klopp2017oracle,gao2015rate,chan2014consistent}. Similarly, the hypergraphon~\citep{zhao2015hypergraph,lovasz2012large} is a limiting function of $m$-uniform hypergraphs whose edges can join $m$-vertices with $m\geq3$. The simple hypergraphon model~\citep{balasubramanian2021nonparametric} is a special case of hypergraphon that considers $m$-multivariate latent functions. Specifically, the simple hypergraphon model assumes the signal tensor $\Theta$ has the form 
\begin{align}
  \Theta(i_1,\ldots,i_m) = f(a_{i_1},\ldots,a_{i_m}) \ \text{ for all } i_1<\ldots < i_m,
\end{align}
where $a_i\in [0,1]$ is a scalar-valued latent variable, and $f\colon[0,1]^m\rightarrow [0,1]$ is a function called the simple hypergraphon. Assuming $f$ is analytic, we see that the simple hypergraphon model belongs to our latent variable tensor model with $(1,\ldots,1)$-latent dimension.
\end{example}

\subsection{The first main result: latent variable tensors are of log-rank}\label{sec:first} 
In this section, we study the approximation theory of high-rank latent variable tensors. A notable result is that we can approximate a latent variable tensor, up to a small \emph{entrywise} error, by a log-rank tensor. 


\begin{thm}[Dimension-free uniform approximation]\label{thm:rank}
For every tensor $\Theta\in\tP{(\md,\ms,M)}$ and every integer $r\in\mathbb{R}_{+}$, there exists an $(r,\ldots,r)$-rank tensor $\tX$ such that 
\begin{align}\label{eq:fboundorg}
    \|\Theta - \tX\|_\infty\leq e^{-c(\ms, M) r^{1/\bar s}},
\end{align}
where $c(\ms, M)>0$ is a constant not depending on the tensor dimension. 
\end{thm}

Theorem~\ref{thm:rank} shows that the entrywise error decays exponentially fast with respect to the approximation rank. The entrywise error is useful in modern data applications. For example, one often wants to approximate an extremely large tensor while perturbing each entry as little as possible; this is exactly captured by the infinity norm in~\eqref{eq:fboundorg}.

We obtain the following proposition based on the Theorem~\ref{thm:rank}.
\begin{prop}[Latent variable tensors are of log-rank]\label{cor:rank} Fix an arbitrary $\varepsilon>0$, and consider $\Theta\in\tP{(\md,\ms,M)}$. As $\bar d\to \infty$, we have
\[
\min\{r\in\mathbb{N}_{+}\colon \text{Tucker-rank}(\tX)\leq (r,\ldots,r), \ \FnormSize{}{\Theta-\tX}\leq \varepsilon\}\lesssim \log^{\bar s}(\bar d),
\]
where the asymptotic notion $\lesssim$ absorbs the constant not depending on the tensor dimension. 
 
\end{prop}

Proposition~\ref{cor:rank} shows that log-rank tensors are enough to approximate latent variable tensors up to small errors. This result explains the empirical success of low-rank based approaches despite the prevalence of high rank tensors. Although tensors generated in real world are often high rank, many of them can be well approximated by log-rank tensors. 

\begin{rmk}[Comparison with previous works on matrices]
Theorem~\ref{thm:rank} is developed by a delicate combination of nonparametric tools, tensor algebra, and functional analysis. The analytic function in \eqref{eq:analytic} is one extension of univariate analytic functions to multivariate cases. Past works~\citep{udell2019big,xu2018rates} introduce an alternative definition of multivariate analytic functions: a multivariate function $f$ is called analytic, if there exists a constant $M>0$ satisfying
\begin{align}\label{eq:analytic2}
    \sup_{\mx_k\in B(\mathbb{R}^{s_1},\|\cdot\|_\infty), k\in[m]}\left|{\partial^{|\boldsymbol{\alpha|}} f(\mx_1,\mx_2, \ldots,\mx_m)\over \partial^{\boldsymbol{\alpha}} \mx_1}\right|\leq M^{|\boldsymbol{\alpha}|}\boldsymbol{\alpha}!,
\end{align}
for all multi-indices $\boldsymbol{\alpha}\in\mathbb{N}^m$. 
As a key difference, the function~\eqref{eq:analytic2} is partially analytic in the first mode only, whereas our function~\eqref{eq:analytic} is jointly analytic on all $m$ modes. The past works~\citep{udell2019big,xu2018rates} use~\eqref{eq:analytic2} to show that latent variable matrices (i.e. the special case $m=2$ in our model) are well approximated by log-rank matrices. 

However, direct applications of past techniques~\citep{udell2019big,xu2018rates} to tensors turn out to be impossible. In fact, we show that the analytic function in \eqref{eq:analytic2} guarantees only the log-rank on the first mode, while the ranks on other modes can be full. For matrices, the log-rank on the first mode is enough, because the column rank and the row rank are the same. For higher-order tensors, however, the log-rank on a certain mode does not guarantee the log-rank of other modes. We propose the joint analytic functions~\eqref{eq:analytic} to ensure the log-rank on all modes as stated in Proposition~\ref{cor:rank}. This highlights the challenges of higher-order tensors compared to matrices.
\end{rmk}

\section{Statistically optimal rate via least-square estimation}\label{sec:lse}

In this section, we first propose the least-square estimator (LSE) and develop its convergence rate. We then present the matching lower bound of the estimation problem. We are interested in the asymptotic regime as $\md \to \infty$ while treating model configuration $(\ms,M)$ fixed. 
 
We propose a rank-constrained least-square estimator for $\Theta$,
\begin{align}\label{eq:lse}
    \hat\Theta^{\textup{LSE}} = \argmin_{\text{Tucker-rank}(\Theta)\leq (r_1,\ldots,r_k)}\FnormSize{}{\tY-\Theta}.
\end{align}
The least-square estimator $\hat\Theta^{\textup{LSE}}$ depends on the rank $(r_1,\ldots,r_k)$ in the constraint. We choose the log-rank as Proposition~\ref{thm:rank} suggested. Our Theorem~\ref{thm:lseupper} establishes the convergence rate of the least-square estimator.

\begin{thm}[Statistical upper bound]\label{thm:lseupper}
 Consider a data tensor $\tY$ generated from~\eqref{eq:gmodel} with the signal tensor $\Theta\in\tP(\md,\ms,M)$. 
Let $\hat\Theta^{\textup{LSE}}$ denote the least-square estimator in \eqref{eq:lse} with $r=r_1=\cdots=r_m=c(\ms, M)^{-\bar s}\log^{\bar s}\left(d_* / \bar d\right)$ for all $k\in[m]$, where $c(\ms, M)$ is the constant in Theorem~\ref{thm:rank}. Then, with high probability, we have
\begin{align}\label{eq:lseupper}
    \FnormSize{}{\hat\Theta^{\textup{LSE}}-\Theta}^2&\lesssim \KeepStyleUnderBrace{r^m+r \bar d}_{\text{low-rank estimation error}} +\KeepStyleUnderBrace{d_*\exp(-c(\ms,M)r^{1/s}) }_{\text{high-rank approximation error}} = \tilde O(\bar d).
\end{align}
\end{thm}

We discuss the implication of the result~\eqref{eq:lseupper}. The LSE rate consists of two sources of error: the estimation error and the approximation error. The former is due to the noise in the observation model, and the latter is due to high-rank structure of the signal. These two terms reflect the bias-variance tradeoff in nonparametric problems. Setting a high fitted rank reduces the approximation error, but increases the estimation error due to more parameters to estimate. Our log-rank approximation theory achieves the trade-off between the two. 

For classical low-rank models, our Theorem~\ref{thm:lseupper} improves earlier work by relaxing the choice of fitted rank. Earlier work requires the fitted rank $(r_1,\ldots,r_k)$ to be precisely the same as the true signal rank~\citep{zhang2018tensor}; whereas our result allows slightly over-fitted rank in the estimation. We find that this over-fitting incurs only a negligible log factor, compared to the rate $\tO(\bar d)$ for low-rank tensors~\citep{zhang2018tensor}.

We now show that the rate \eqref{eq:lseupper} cannot be improved, up to a log factor, by any estimation methods. The following theorem establishes the lower bound of the estimation problem. 

\begin{thm}[Statistical lower bound]\label{thm:lselower}
Consider the same set-up as in Theorem~\ref{thm:lseupper}. Let $\hat\Theta$ be any estimator based on the observed data tensor $\tY$. Then, there exists an absolute constant $p_0\in(0,1)$ such that
\begin{align}\label{eq:lselower}
    \inf_{\hat\Theta}\sup_{\Theta\in\tP(\md,\ms,M)}\mathbb{P}\left(\FnormSize{}{\hat\Theta-\Theta}^2\gtrsim \bar d \right)\geq p_0.
\end{align}
\end{thm}

Theorem~\ref{thm:lselower} guarantees that, there exists no estimator with a convergence rate faster than $\tO(\bar d)$. The statistical lower bound is obtained via information theoretical analysis. This lower bound applies to all estimators obtained from the observed tensor $\tY$, including but not limited to the least-square estimator (LSE) and the double-projection spectral estimator (DSE) introduced in later sections. The match between the lower bound in \eqref{eq:lselower} and the upper bound in \eqref{eq:lseupper} implies the statistical optimality of LSE in \eqref{eq:lse}.




\section{Computationally optimal rate via double-projection spectral estimation}\label{sec:comp}
Unlike the matrices, solving the rank constrained least-square problem in \eqref{eq:lse} is NP-hard and numerically ill-posed in general~\citep{hillar2013most}.
One question is whether the statistically optimal rate can be achieved with an estimator computable in polynomial time. In this section, we answer this question. We first construct the computational lower bound based on the hypergraphic planted clique problem. Then, we propose a polynomial-time double-projection spectral estimation that achieves the computationally optimal rate. The result guarantees no polynomial-time algorithms achieve the statistically optimal rate, thereby revealing a non-avoidable statistical-computational gap for tensors of order $m\geq 3$. 

\subsection{Hypergraphic planted clique detection}
The hypergraphic planted clique detection plays an important role in constructing the computational lower bound of our problem. In this section, we briefly introduce the hypergraphic planted clique (HPC) model and HPC conjecture.

Consider an $m$-uniform hypergraph  $G = (V,E)$, where $V$ is a set of vertices, and $E$ is a set of hyperedges of the hypergraph. We define an adjacency tensor $\tA = \tA(G)$  corresponding to the hypergraph $G$ as
\begin{align}
    \tA(i_1,\ldots,i_m) = \begin{cases}1,&\text{ if } (i_1,\ldots,i_m)\in E,\\ 0, & \text{ otherwise. }\end{cases}
\end{align}
The Erd\"os-R\'enyi random hypergraph with $d$ vertices, denoted as $\tG_m(d,1/2)$, is a random $m$-uniform hypergraph of which the probability of each hyperedge connection is 1/2.
The hypergraphic planted clique (HPC) with clique size $\tau>0$, denoted as $\tG_m(d,1/2,\tau)$, is generated from an Erd\"os-R\'enyi random hypergraph in the following way. First we generate an Erd\"os-R\'enyi random hypergraph from $G_m(d,1/2)$. Then, we independently pick $\tau$ vertices with equal probability from $d$ vertices. Finally, we obtain an HPC by including only the hyperedges whose vertices all belong to the picked $\tau$ vertices.

The HPC detection refers to the following hypothesis testing problem,
\begin{align}\label{eq:HPC}
    H_0\colon G\sim \tG_m(d,1/2)\quad\text{v.s.}\quad H_1\colon G\sim \tG_m(d,1/2,\tau).
\end{align}
Given an adjacency tensor $\tA(G)$ and a test $\phi$ for \eqref{eq:HPC}, the performance of the test is evaluated by the sum of Type-I and II errors, i.e.,
\begin{align}\label{eq:err}
\text{Err}(\phi)=\mathbb{P}_{H_0}(\phi(\tA) = 1)+ \mathbb{P}_{H_1}(\phi(\tA) = 0).
\end{align}
The HPC detection problem is a generalization of the well studied planted clique (PC) detection for the grapes (i.e., $m = 2$ in our setting).  For graphs with $m=2$, the famous planted clique conjecture has been extensively used as a computational hardness assumption in statistical problems~\citep{berthet2020statistical,wang2016statistical,cai2020statistical}.
 Similar to the hardness conjecture for the PC detection, early work~\citep{luo2022tensor} presents the hardness conjecture for the HPC detection. In the next section, we construct the computational lower bound of the latent variable tensor estimation based on the following hardness conjecture.
\begin{conj}[HPC detection conjecture \citep{luo2022tensor}]\label{conj:1} Consider the HPC problem in \eqref{eq:HPC} and suppose $m\geq 2$ is a fixed integer. If
\begin{align}
    \limsup_{d\rightarrow\infty} \log\tau/\log\sqrt{d}\leq 1-\varepsilon,\text{ for any } \varepsilon >0.
\end{align}
Then, for any polynomial-time test sequence $\{\phi\}_d\colon\tA\mapsto \{0,1\}$, we have 
\begin{align}
   \liminf_{d\rightarrow\infty} \text{Err}(\phi)\geq \frac{1}{2}.
\end{align}
\end{conj}

\subsection{Computational lower bound under HPC detection conjecture}
Now we are ready to present the computational lower bound for the high-rank tensor estimation problem based on Conjecture~\ref{conj:1}. 

\begin{thm}[Computational lower bound]\label{thm:polylower} Assume Conjecture~\ref{conj:1} holds. Consider the same set-up as in Theorem~\ref{thm:lseupper}. Then for every polynomial-time computable estimator $\hat\Theta$, we have, as $\underline d\to\infty$,
\begin{align}
    \frac{1}{d_*^{1/2-\epsilon}}\sup_{\Theta\in\tP(\md,\ms,M)}\mathbb{E}\FnormSize{}{\hat\Theta-\Theta}^2\rightarrow \infty,\quad \text{for any  }\epsilon >0.
\end{align}

\end{thm}
Theorem~\ref{thm:polylower} presents the fundamental limit of computationally feasible estimators. The result implies that, there exists no polynomial-time estimator with a convergence rate faster than $\tO(d_*^{1/2})$. 
The next section shows that this computational lower bound is attainable, up to logarithm factors, by our double-projection spectral estimator.

\subsection{Computational upper bound via double-projection spectral algorithm}
We now propose an efficient double-projection spectral algorithm that achieves the computational lower bound in Theorem~\ref{thm:polylower}. As a consequence, our double-projection spectral estimator (DSE) is computationally optimal.

The main idea of our DSE is to combine the log-rank approximation and matrix spectral decomposition upon projection. Given an approximation rank $(r_1,\ldots,r_k)$, we apply singular value decomposition (SVD) \emph{twice} on the unfolded tensor. The first SVD is on the unfolded observed tensor $\tY$, i.e.,
\begin{align}
    \tilde\mU_k = \textup{SVD}_{r_k}\left(\textup{Unfold}_k(\tY)\right), \quad k\in[m],
\end{align}
where $\textup{SVD}_{r_k}(\cdot)$ returns the top-$r_k$ left singular vectors of the matrix. The second SVD is on the unfolded tensor after projection onto pre-estimated subspaces of the other $(m-1)$ modes, i.e.,
\begin{align}\label{eq:second}
    \hat\mU_k =\textup{SVD}_{r_k}\left(\textup{Unfold}_k\left(\tY\times_1\tilde \mU_1^T\times_2 \cdots \times_{k-1}\tilde \mU_{k-1}^T\times_{k+1}\tilde \mU_{k+1}^T\times \cdots \times_m \tilde \mU_m^T\right)\right),\quad k\in[m].
\end{align}
Finally, we estimate the signal tensor $\Theta$ by
\begin{align}\label{eq:polyest}
    \hat\Theta^{\text{DSE}} = \tY\times_1(\hat \mU_1\hat \mU_1^T)\times \cdots\times_m (\hat \mU_m\hat \mU_m^T).
\end{align}

The full procedure of DSE is summarized in Algorithm~\ref{alg:polyalg}. The estimation involves matrix operations only and thus is polynomial-time computable. Our algorithm differs from tensor higher-order SVD (HOSVD)~\citep{de2000multilinear} and higher-order orthogonal iteration (HOOI)~\citep{zhang2018tensor}. The HOSVD projects the observed tensor $\tY$ only \emph{once} by using $\{\tilde \mU_k\}_{k=1}^m$ in \eqref{eq:polyest}. The HOOI~\citep{zhang2018tensor} performs the projection repeatedly in $\tO(\log p)$ iterations.  We find that our double-projection algorithm improves HOSVD and reduces the number of iteration from HOOI, by alleviating the noise effects upon unfolding~\citep{han2022exact} with two projections.

\begin{algorithm}[h]
  \caption{Double-projection spectral algorithm}\label{alg:polyalg}
 \begin{algorithmic}[1] 
\INPUT A noisy data tensor $\tY\in\mathbb{R}^{d_1\times \cdots\times d_m}$ and the approximation rank $(r_1,\ldots,r_k)$. 
\State Compute  $\tilde\mU_k = \textup{SVD}_{r_k}\left(\textup{Unfold}_k(\tY)\right)$ for $k = 1,\ldots, m.$
\ForAll{$k = 1,\ldots, m$}
\State Compute the singular space estimate
\begin{align}
    \hat\mU_k =\textup{SVD}_{r_k}\left(\textup{Unfold}_k\left(\tY\times_1\tilde \mU_1^T\times \cdots \times_{k-1}\tilde \mU_{k-1}^T\times_{k+1}\tilde \mU_{k+1}^T\times \cdots \times_m \tilde \mU_m^T\right)\right).
\end{align}
\EndFor
\State Estimate the signal tensor by $ \hat\Theta^{\text{DSE}} = \tY\times_1(\hat \mU_1\hat \mU_1^T)\times \cdots\times_m (\hat \mU_m\hat \mU_m^T).$
\OUTPUT Estimated signal tensor $\hat\Theta^{\text{DSE}}$.
\end{algorithmic}
\end{algorithm}

We now establish the statistical accuracy of our DSE. 

\begin{thm}[Computational upper bound via DSE algorithm]\label{thm:polyupper} Consider the same set-up as in Theorem~\ref{thm:lseupper}. 
Let $\hat\Theta$ be the estimator obtained from Algorithm~\ref{alg:polyalg} with the input tensor $\tY$ and the approximation rank $r_k = c(\ms, M)^{-\bar s}\log^{\bar s}\left (d_*/ \bar d\right)$ for all $k\in[m]$.
Then, with high probability, we have
\begin{align}\label{eq:polyupper}
    \FnormSize{}{\hat\Theta^{\text{DSE}}-\Theta}^2 &\lesssim d_*^{1/2}\log ^{\bar s}\left(d_*\over \bar d\right) + \bar d\log^{2\bar s} \left(d_*\over \bar d\right), \\&= \tilde \tO(d_*^{1/2}\vee\bar d).
\end{align}
\end{thm}

We see that the upper bound in \eqref{eq:polyupper} matches with the computational lower bound in Theorem~\ref{thm:polylower} up to logarithm factors. Therefore, our estimator $\hat \Theta^{\text{DSE}}$ is computationally optimal.

\begin{rmk}[Gap between statistical and computational optimality]
Theorems~\ref{thm:polylower}-\ref{thm:polyupper} show that the computationally optimal rate for estimating the signal tensor $\Theta\in\tP(\md,\ms,M)$  is  of order $\tilde \tO(d_*^{1/2})$. Theorems~\ref{thm:lseupper}-\ref{thm:lselower} show that the statistically optimal rate is of order $\tilde \tO(\bar d)$, among all estimates including those computationally intractable. For easier comparison between the two rates, we consider the equal dimension as $d_1 = \cdots = d_m = d$.  Figure~\ref{fig:semantic} shows the statistical-computational gap with respect to mean square error  
\[
\text{MSE}(\hat \Theta, \Theta)={1\over \prod_{k\in[m]} d_k}\FnormSize{}{\hat \Theta-\Theta}^2.
\]
We see that the statistically optimal MSE is $\tilde \tO(d^{-(m-1)})$, whereas the computationally optimal MSE is $\tilde \tO(d^{-m/2})$. The gap in between is called the ``statistically possibly but computationally impossible'' region. Our two estimators, LSE and DSE, are located at the two boundaries respectively. The statistical-computational gap becomes more noticeable as the tensor order $m$ increases, while no gap exists for the matrices with $m=2$. This phenomenon reveals the fundamental challenges with higher-order tensors. 
\end{rmk}

\begin{rmk}[Necessity of double projection procedure]
For matrices (i.e., $m=2$), the second projection~\eqref{eq:second} becomes redundant because $\tilde \mU_k$ and $\hat \mU_k$ are the same for $k = 1,2.$ The projection of the observed matrix $\tY\in\mathbb{R}^{d_1\times d_2}$ onto the best rank-$r$ column space is equal to the best rank-$r$ row space. The step~\eqref{eq:polyest} reduces to the truncated-SVD estimator, and its accuracy is guaranteed by the Eckart-Young Theorem~\citep{eckart1936approximation}. 

For higher-order tensors, however, the projection~\eqref{eq:second} onto the subspace at one mode changes the subspaces at other modes. In fact, computing the best rank-$r$ subspaces of $\tY$ jointly at all modes is NP-hard~\citep{hillar2013most}. We have to control the noise aggregation of subspace estimation in a more complicated way, and the additional projection of $\tY$ onto the pre-estimated subspaces is needed~\citep{han2022exact}.
We introduce the double projection procedure to overcome these challenges. This additional projection of $\tY$ on the pre-estimated subspaces substantially reduces the noise effect in $\hat\mU_k$. We verify the benefit of the double projection over traditional HOSVD in Section~\ref{sec:sim}. 
\end{rmk}

\begin{rmk}[Hyperparameter tuning]
Our algorithm~\ref{alg:polyalg} treats the approximation rank as inputs. The theoretical choice of the rank is given in  Theorem~\ref{thm:polyupper}. In practice, since the model configuration is unknown, we need to search the approximation rank as tuning parameters. Based on our simulations, we find it sufficient to set $r_k = r=c\log \bar d$ for all $k\in[m]$ and choose the constant $c$ via cross-validation. We investigate the practical impacts of tuning parameters in both synthetic and real data applications in Section~\ref{sec:sim}.
\end{rmk}

\section{Numerical analysis}\label{sec:sim}
We study the performance of our methods through simulations and real data applications.

\subsection{Simulations}
We first verify our theoretical results using synthetic data. We simulate order-3 tensors based on the latent variable tensor model \eqref{eq:LVM}. We sample independent latent variables $\ma_i\sim\textup{Unif}[0,1]^s$ for $i\in[d]$, and generate the signal tensor $\Theta\in\mathbb{R}^{d\times d \times d}$ based on 
\begin{align}\label{eq:sim}
    \Theta(i_1,i_2,i_3) = f(\ma_{i_1},\ma_{i_2},\ma_{i_3}) \text{ for all } (i_1,i_2,i_3)\in[d]^3,
\end{align}
where $f\colon\mathbb{R}^s\times\mathbb{R}^s\times\mathbb{R}^s\rightarrow\mathbb{R}$ is an analytic function. We consider three simulation models listed in Table~\ref{tb:model}. The functions involve compositions of operations such as polynomial, logarithm, exponential, and cosine. Notably, the signal tensors generated by these analytic functions are all full rank.
\begin{table}[H]
    \centering
    \caption{Analytic functions in \eqref{eq:sim} to generate signal tensors.}
    \begin{tabular}{c|c}
    \hline
        Model ID &  $f(\mx,\my,\mz)$ \\\hline
        1 &  $\exp\left(-(\|\mx-\my\|_2^2+\|\my-\mz\|_2^2+\|\mz-\mx\|_2^2)/{3}\right)$\\
        2& $\cos\left((\|\mx-\my\|_2^2+\|\my-\mz\|_2^2+\|\mz-\mx\|_2^2)/{3}\right)$\\
        3 & $\log\left(1+ (\|\mx-\my\|_2^2+\|\my-\mz\|_2^2+\|\mz-\mx\|_2^2)/{3}\right)$\\\hline\hline
    \end{tabular}
    \label{tb:model}
\end{table}

The first experiment examines the numerical log-rank of the generated signal tensors from \eqref{eq:sim} as suggested in our Proposition~\ref{cor:rank}.
Our theory shows that a tensor generated from the latent variable tensor model \eqref{eq:LVM} is well approximated by a log-rank tensor.
We define the numerical $\epsilon$-rank of a tensor $\Theta$ as
\[\textup{rank}_\epsilon(\Theta) = \min\{r\in\mathbb{N}_+\colon \textup{rank}(\tX)\leq(r,r,r),\FnormSize{}{\Theta-\tX}\leq \epsilon\FnormSize{}{\Theta}\}.\]
Computing the exact $\epsilon$-rank for $0<\epsilon<1$ is NP hard. A common approach to approximate $\textup{rank}_\epsilon(\Theta)$ is to project $\Theta$ onto an $r$-rank tensor by the HOOI algorithm~\citep{zhang2018tensor} and check the approximation error. We set the $\epsilon$-rank using set the smallest rank which makes the relative approximation error below $\epsilon$. We vary the dimension of latent variables $s\in\{5,10,15\}$ and tensor dimension $d\in\{20,40,\ldots,200\}$ with $\epsilon = 0.01.$  

Figure~\ref{fig:lvmrank} plots the numerical rank versus tensor dimensions under each of the three simulation models. We find that the numerical rank appears to grow on the order of logarithm of tensor dimension~$d$. This trend verifies our log-rank approximation theory of high rank tensors. In addition,  we observe an upward shift of the curve as latent dimension $s$ increases. This phenomenon is consistent with Proposition~\ref{cor:rank}, which says the $\varepsilon$-rank is upper-bounded by $\tO(\log^s(d))$. Intuitively, a higher latent dimension implies a higher model complexity, and thus a higher approximation rank is needed. 

\begin{figure}[h]
    \centering
    \includegraphics[width = \textwidth]{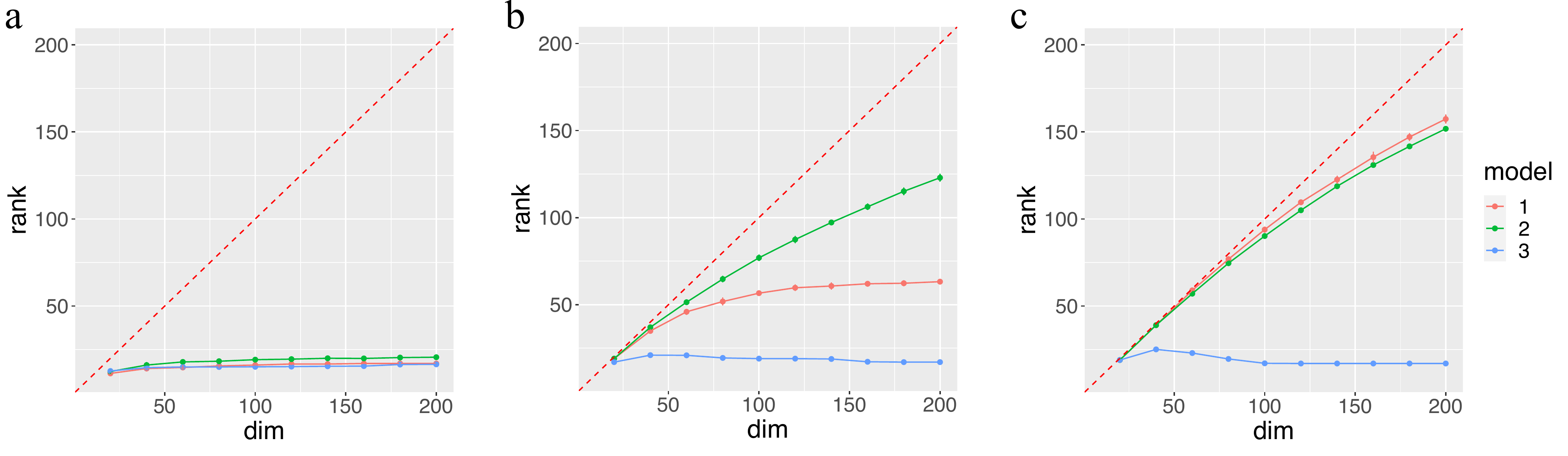}
    \caption{Numerical rank of $\Theta$ versus tensor dimension $d$ under three models in Table~\ref{tb:model}. Panels a-c correspond to the latent dimension $s\in\{5,10,15\}$, respectively.}
    \label{fig:lvmrank}
\end{figure}

         
         
         

The second experiment investigates the performance of our DSE algorithm for signal estimation. We generate noisy tensors $\tY = \Theta + \tE$ from \eqref{eq:sim}, where $\Theta$ is the signal tensor and $\tE$ is the noise tensor with i.i.d. entries from $N(0,1)$. We estimate the signal tensor using the DSE algorithm with input tensor $\tY$ and the approximation rank $r$. The performance is evaluated using mean square error. We set the approximation rank as $r = c\log^s d$ based on Theorem~\ref{thm:polyupper}, where we choose $c$ that gives the best result.

Figure~\ref{fig:est} shows the mean squared error versus tensor dimension under latent dimensions $s\in\{1,2,3\}.$ We see that the DSE algorithm successfully estimates the signal tensors in all scenarios. The MSE shows a polynomial-decaying trend with respect to the tensor dimension $d$. This result again is consistent with our Theorem~\ref{thm:polyupper}. In addition, we find that the decaying trends are similar for various $s$, implying the robustness of our algorithm again latent dimensions. 

\begin{figure}[h]
    \centering
    \includegraphics[width = \textwidth]{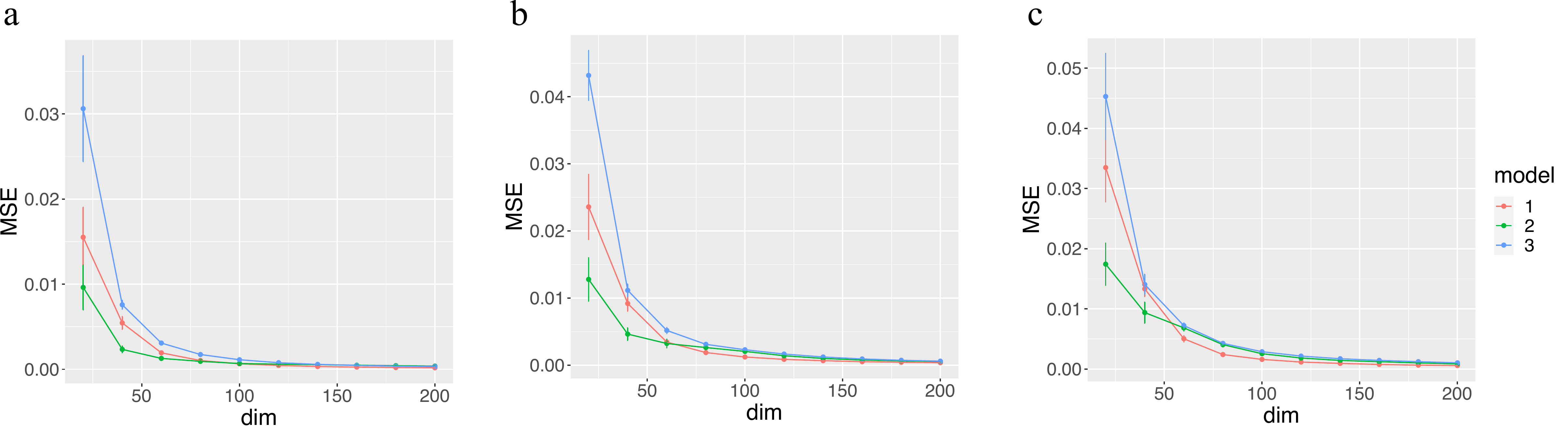}
    \caption{MSE versus tensor dimension $d$ under three models in Table~\ref{tb:model}. For each setting, we plot the average MSE and its standard error across 20 replicates.  Panels a-c correspond to the latent dimension $s\in\{1,2,3\}$, respectively.}
    \label{fig:est}
\end{figure}

Lastly, we compare our DSE algorithm with two other popular tensor estimation algorithms: HOSVD~\citep{de2000multilinear} and Borda Count~\citep{lee2021smooth}. The HOSVD uses the single projection whereas our DSE algorithm takes the double projection procedure. The Borda Count algorithm assumes the monotonic assumption on the latent function and uses sorting and smoothing for estimation. Figure~\ref{fig:svsd} plots the MSE of HOSVD, Borda Count, and DSE algorithms. We fix $s=2$ and vary $d\in\{20,40,...,200\}$. We see that our DSE algorithm achieves the best performance in most scenarios, whereas the Borda Count is often the worst. The underperformance of Borda Count algorithm is possibly due to the lack of monotonicity in the latent functions. 
We also find that the double projection consistently improves the performance of the single projection, from the fact that DSE always outperforms HOSVD. This result verifies our intuition that the double projection alleviates the noise effects for higher-order tensors.

\begin{figure}[h]
    \centering
    \includegraphics[width = \textwidth]{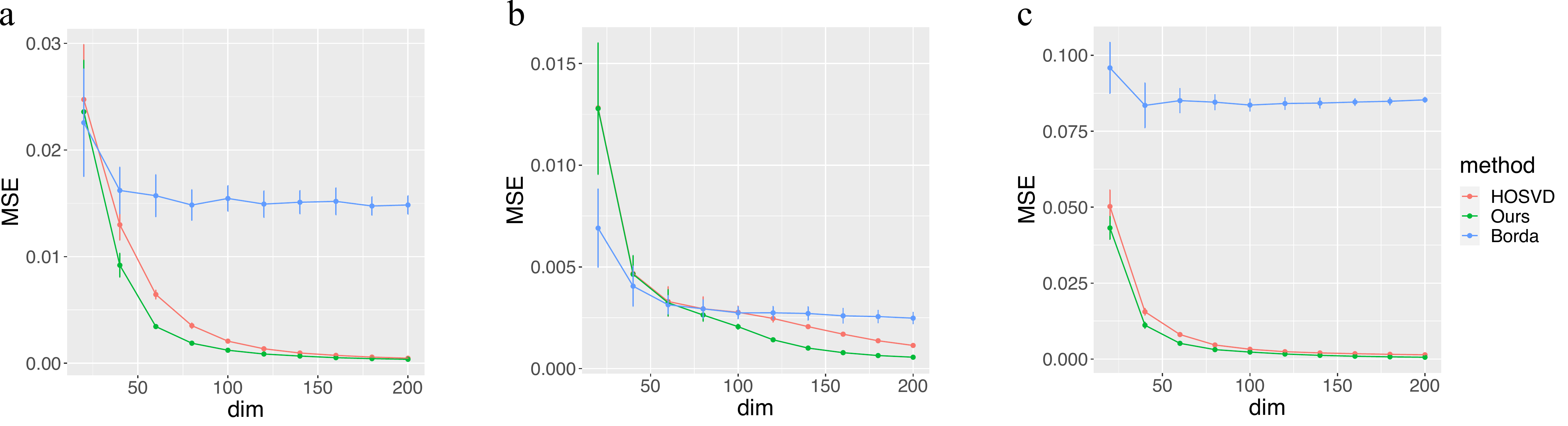}
    \caption{MSE versus tensor dimension $d$ across different models from Table~\ref{tb:model} and estimation algorithms (HOSVD, Borda Count, and Ours). Panels a-c corresponds to simulation models 1-3 in Table~\ref{tb:model}.}
    \label{fig:svsd}
\end{figure}

\subsection{Application to fMRI 3D brain image}\label{sec:brain}
We apply the DSE algorithm to visual motion fMRI database~\citep{buchel1997modulation}. The fMRI brain image consists of voxels, equivalent to 3D pixcels, across length, width, and height. The observed 3D brain image tensor $\Theta$ is of size $157\times 189\times 68$ and full rank. Figure~\ref{fig:brain1}(a) plots a horizontal slice of the tensor brain image, and Figure~\ref{fig:brain2}(a) shows the original 3D image across vertical and horizontal lines for better visualization.

In this study, we add i.i.d.\ Gaussian noise $N(0,\sigma_\gamma^2)$ to each entry of the original tensor, where
\begin{align}
    \sigma_\gamma = \gamma\left(\frac{\FnormSize{}{\Theta}^2}{157\times 189\times 68}\right)^{1/2},\quad \text{with noise level }\gamma\in\{0,0.2,0.4,\ldots,0.8,1\}.
\end{align}
Notice that the noise level $\gamma = 0$ means the original signal tensor without the noise. 
We apply the DSE algorithm to the contaminated image tensor by varying the approximation rank $r\in\{3,6,\ldots,60\}$.  Figure~\ref{fig:brain1}(b) plots the MSE versus the approximation rank under different noise levels $\gamma$. 
In the absence of noise, setting a higher rank results in a smaller approximation error, as expected from our Theorem~\ref{thm:lseupper}. In the presence of noise, however, setting a higher rank results in a suboptimal bias-variance tradeoff. We find that the best performance is often achieved at the rank 20 -- 40, with a higher rank corresponding to the lower noise. 

\begin{figure}[H]
     \centering
    \includegraphics[width = \textwidth]{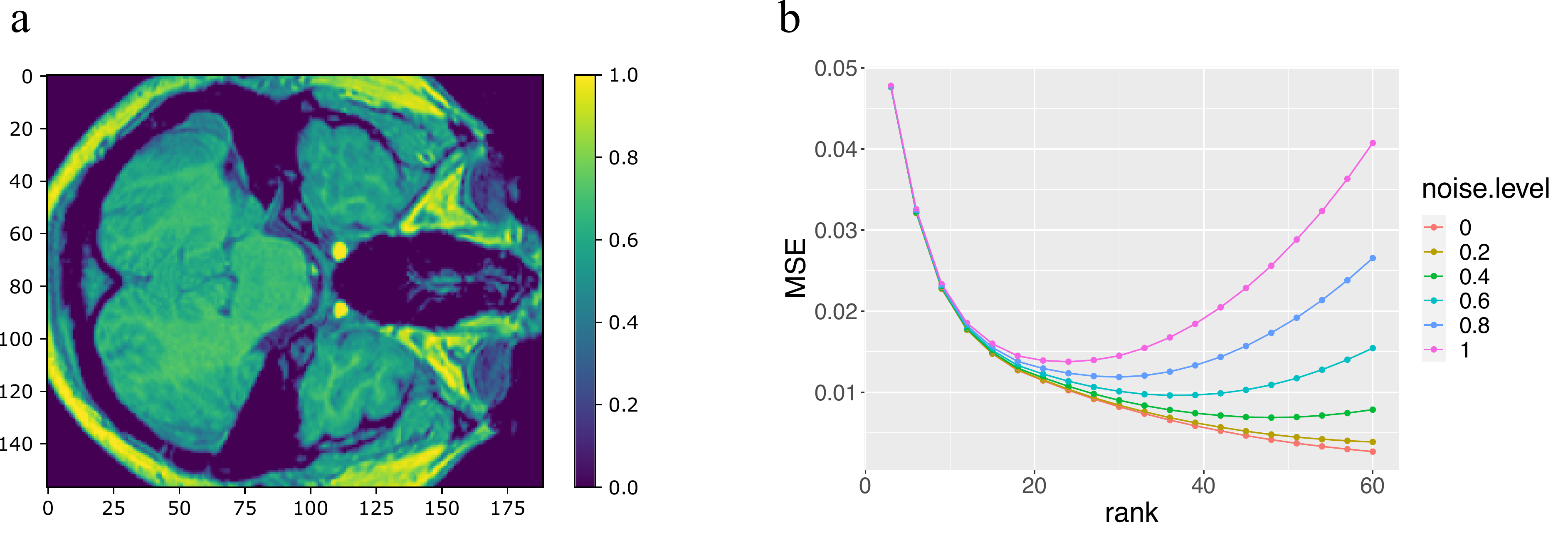}
        \caption{(a): A horizontal slice of tensor brain image $\{\Theta(i,j,5):(i,j)\in[157]\times[189]\}$. (b): MSE of the DSE algorithm under various rank approximations and noise levels. }
        \label{fig:brain1}
\end{figure}
Figure~\ref{fig:brain2} displays the input brain image data tensor with noise levels $\gamma\in\{0,0.5,1\}$, and the corresponding denoised tensor from our algorithm with the approximation rank $r = 27$. We see that the DSE algorithm successfully recovers the original brain image in all scenarios. \begin{figure}[h]
     \centering
    \includegraphics[width = \textwidth]{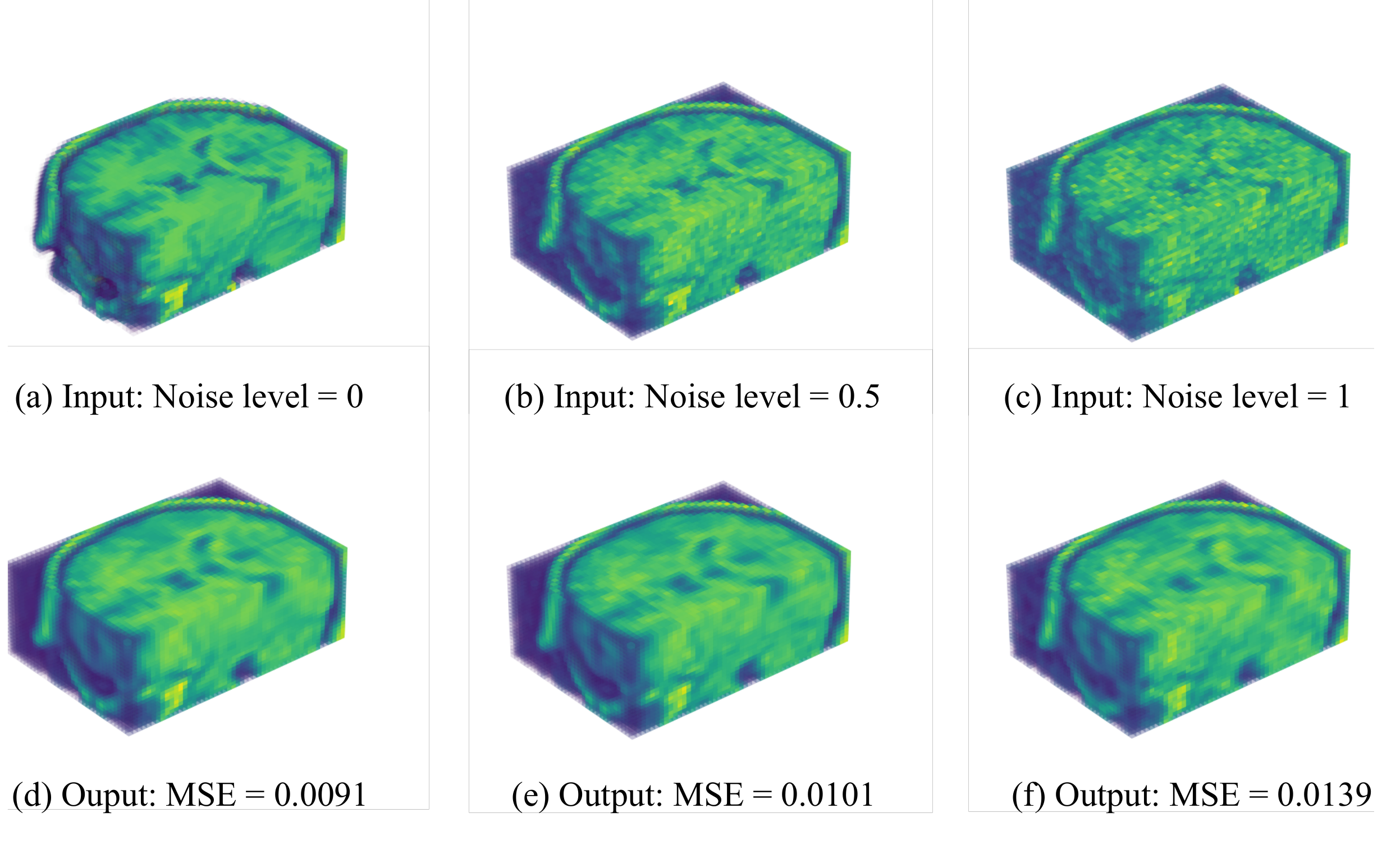}
        \caption{Panels (a)-(c) plot the input tensors with three noise levels $\gamma\in\{0,0.5,1\}$. Panels (d)-(f) plot the denoised tensors from input tensors (a)-(c), respectively. }
        \label{fig:brain2}
\end{figure}

\begin{figure}[h]
    \centering
    \includegraphics[width = \textwidth]{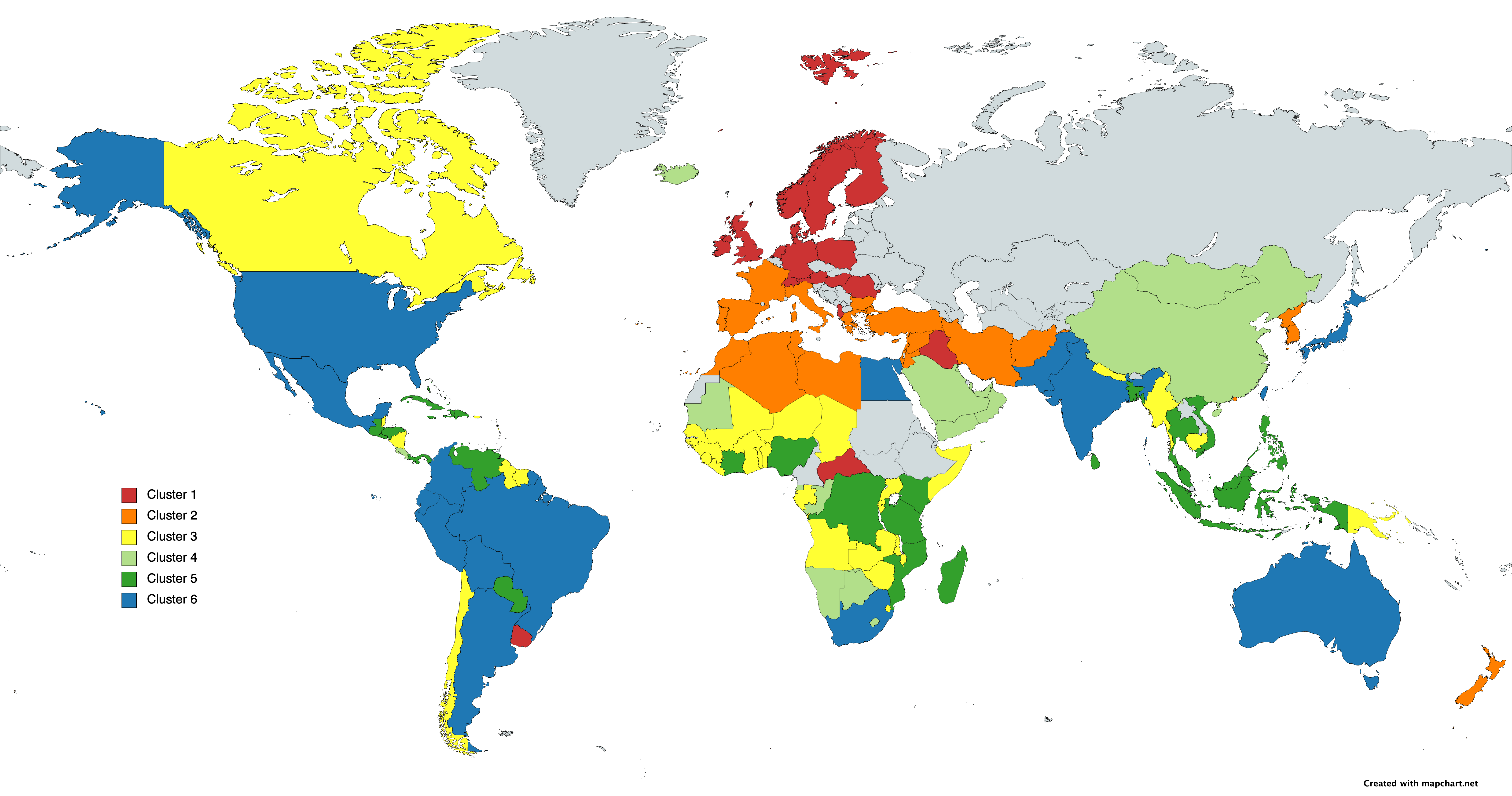}
    \caption{Six clusters of 166 countries based on our estimated signal tensor obtained from DSE algorithm in the crop production data analysis.}
    \label{fig:cropworld}
\end{figure}

\subsection{Application to crop production data}\label{sec:crop}
We apply the DSE algorithm to the world-wide crops production data  as the second illustration (available on \href{http://www.fao.org/faostat/en/
#data/QC}{http://www.fao.org/faostat/en/\#data/QC}).  The dataset contains annual production of 161 crops  across 166 countries between 1961 and 2020. The observed data can be organized as an order-3 tensor with entries representing the log counts of productions from 161 crops, 166 countries, and 60 years. We select the approximation rank $r= 26$ via cross-validation and obtain the denoised tensor from the DSE algorithm. We perform a clustering analysis via $K$-means on the country and crop modes based on the estimated signal tensor (see detailed procedure in Appendix). Figure~\ref{fig:cropworld} shows that our clustering results successfully capture the regional separation. We find that the Cluster 1 represents the western Europe, whereas the Cluster 2 consists of countries around Mediterranean sea. The Cluster 3 mainly contains countries near African deserts, and the Cluster 4 includes east and west Asia. The Cluster 5 represents south east Asia and east Africa regions, while the Cluster 6 mainly consists of countries around Pacific ocean. We emphasize that our clustering is consistent with actual country locations without using geographic information such as longitude or latitude as inputs. 
Similar clustering analysis is performed on crop mode, and we find that the clustered groups capture the type of crops. The detailed procedures and results are provided in Appendix.
The identified similarities among countries and crops within same clusters illustrate the applicability of our methods.

\section{Proofs}\label{sec:proof}

\subsection{Proofs of Theorem~\ref{thm:rank} and Proposition~\ref{cor:rank}}

\begin{proof}[Proof of Theorem~\ref{thm:rank}]
Let $\tB_{N,\ell-1}$ be the basis tensor in Lemma~\ref{lem:lvmrank}.
We show that the desired properties are satisfied when taking 
\begin{equation}\label{eq:choice}
 N=\lceil M^2\exp(2m\bar s)\rceil,\quad \text{and}\quad \ell = r^{1/\bar s} c^{-1/\bar s}N^{-m},
\end{equation}
where $\lceil \cdot \rceil$ denotes the ceiling function. 
To verify this, setting $\ell = r^{1/\bar s} c^{-1/\bar s}N^{-m}$ yields the rank bound,
\[
\text{Tucker-rank}(\tB_{N,\ell-1}) \leq cN^{m\bar s}(\ell^{s_1},\ldots,\ell^{s_k})\leq (r,\ldots,r).
\]
Next, we consider the approximation error. From Lemma~\ref{lem:lvmrank}, the approximation error is bounded by
\begin{align}\label{eq:errorbd}
     \|\Theta-\tB_{N,\ell-1}\|_\infty & \lesssim \left( M\over N\right)^\ell \ell^{m \bar s}\nonumber\\
     &\stackrel{(a)}{\leq} \exp\left(- \ell (\log N- \log M-m \bar s )\right)\\
     & \stackrel{(b)}{\leq} \exp\left(- \ell (\log M+m\bar s)\right) \\
     &\stackrel{(c)}{=} \exp(-r^{1/\bar s} c'(\log M+m\bar s))\\
     &= \exp(-r^{1/\bar s}c'')
\end{align}
where the line (a) uses $\log \ell \leq \ell$, the line (b) uses the fact that $\log N \geq 2(\log M+m\bar s)$, the line (c) plugs in $\ell = r^{1/\bar s}c'$ in~\eqref{eq:choice} with a constant $c'=c'(\ms, M)$, and the last line absorbs the term $(\log M+m\bar s)$ into another constant $c''$. 

Finally, setting $\tX = \tB_{N,\ell-1}$ in the approximation completes the proof.
\end{proof}

\begin{proof}[Proof of Proposition~\ref{cor:rank}]
From Theorem~\ref{thm:rank}, we have 
\begin{align}
    \FnormSize{}{\Theta-\tX}^2\leq d_*\|\Theta-\tX\|_\infty^2\leq \bar d^m e^{-c r^{1/\bar s}}.
\end{align}
Solving the inequality $\bar d^m e^{-cr^{1/\bar s}}\leq \varepsilon^2$ yields $r\geq c^{-\bar s}(m \log^{\bar s}\bar d-2\log^{\bar s}\varepsilon)$.
By absorbing all constants into $c=c(\varepsilon,M,\ms)$, we conclude
\[
\min\{r\in\mathbb{N}_{+}\colon \text{Tucker-rank}(\tX)\leq (r,\ldots,r), \FnormSize{}{\Theta-\tX}\leq \varepsilon\} \leq c\log^{\bar s}\bar d.
\]
\end{proof}

\begin{lem}[Dimension-free approximation for latent variable tensor model]\label{lem:lvmrank}
Consider a tensor $\Theta\in\tP(\md,\ms,M)$. There exist a set of basis tensors $\{\tB_{N,\ell}\in \tP(\md,\ms,M) \colon (N,\ell)\in\mathbb{N}_{+}\times \mathbb{N}_{+}\}$ and a constant $c=c(m,\ms,M)$, such that, for every integer pair $(N,\ell)$, we have
\begin{enumerate}
\item[(i)] [Approximation error] $\|\Theta-\tB_{N,\ell-1}\|_\infty \leq cM^{\ell}N^{-\ell}\ell^{m\bar s}$;
\item[(ii)]  [Low-rankness] $\text{Tucker-rank}(\tB_{N,\ell-1}) \leq c N^{m\bar s}(\ell^{s_1},\ldots,\ell^{s_k})$.
\end{enumerate}
\end{lem}

\begin{proof}[Proof of Lemma~\ref{lem:lvmrank}]
We first construct the tensor $\tB_{N,\ell-1}$ with small approximation error and then prove the low-ranknss of $\tB_{N,\ell-1}$.
For ease of presentation, we will use the following shorthand notation in the proof. For an $m$-dimensional vector $\mx=(x_1,\ldots,x_m)$ and an integer partition $\alpha=\alpha_1+\cdots+\alpha_m$ with $\alpha_i\in\mathbb{N}$, we write
\begin{align}\label{eq:notation}
\mx^{\alpha}:=x_1^{\alpha_1}\cdots x_m^{\alpha_m}, \quad \text{and}\quad \partial \mx^\alpha:=\partial x_1^{\alpha_1}\cdots \partial x_m^{\alpha_m}, 
\end{align}

\paragraph*{Approximation error} Without loss of generality, assume the latent variables $\ma^{(k)}_{i_k} \in [-1,1)^{s_k}$ for all $k\in[m]$. We partition the interval $[-1,1)$ into $(2N)$ equal-sized intervals; i.e., 
\[
[-1,1)=\bigcup_{i=-N+1}^N\left[{i-1\over N},{i\over N}\right).
\]
Similarly, for a given $k\in[m]$, we partition $[-1,1)^{s_k}$ into $s_k$-fold Cartesian products of equal-sized intervals; i.e., 
\[
[-1,1)^{s_k} =\bigcup_{i_{s_k}=-N+1}^N \cdots\bigcup_{i_1=-N+1}^N \left(\bigtimes_{j=1}^{s_k} \left[{i_j-1 \over N},{i_j\over N}\right)\right), 
\]
where $\bigtimes$ denotes the Cartesian product of intervals. Since our domain space for latent variables is $\bigtimes_{k=1}^{m}[-1,1)^{s_k}$, we have in total $N_{\text{tot}} := (2N)^{\sum_{k=1}^ms_k}$ hypercubes in the partition. We denote these hypercubes by $I_1,\ldots,I_{N_{\text{tot}}}$; i.e., 
\begin{equation}\label{eq:partition}
\bigtimes_{k=1}^{m}[-1,1)^{s_k}=\bigcup_{i=1}^{N_{\text{tot}}} I_i, 
\end{equation}
and we denote centers of these hypercubes by
\[
\KeepStyleUnderBrace{(\mz_1^{(1)},\ldots,\mz_1^{(m)})}_{\in\mathbb{R}^{s_1\times \cdots \times s_m}},\ldots,\KeepStyleUnderBrace{(\mz_{N_{\text{tot}}}^{(1)},\ldots,\mz_{N_{\text{tot}}}^{(m)})}_{\in\mathbb{R}^{s_1\times\cdots\times s_m}}.
\]
Based on the partition~\eqref{eq:partition}, we construct a piece-wise polynomial of degree $\ell$ as
\begin{align}\label{eq:piecewise}
    P_{N,\ell}(\mx_1,\ldots,\mx_m) = \sum_{i=1}^{N_{\text{tot}}} P_{I_i,\ell}(\mx_1,\ldots,\mx_m)\mathds{1}\{(\mx_1,\ldots,\mx_m)\in I_i\},
\end{align}
where 
\begin{align}
&P_{I_i,\ell}(\mx_1,\ldots,\mx_m)\\
=&\ \sum_{|\bm \alpha|\leq\ell}\frac{1}{\bm \alpha!}\left((\mx_1,\ldots,\mx_m)-(\mz_i^{(1)},\ldots, \mz_i^{(m)})\right)^{\bm \alpha}{\partial^{|\bm \alpha|} f(\mx_1,\ldots,\mx_m) \over \partial (\mx_1,\ldots,\mx_m)^{\bm \alpha}} \bigg|_{(\mz^{(1)}_i,\ldots,\mz^{(m)}_i)},
\end{align}
and we have used the shorthand notation~\eqref{eq:notation} in the summands.

By the definition of analytic functions \eqref{eq:analytic}, for all $(i_1,\ldots,i_m)$, we have the following approximation error
\begin{align}\label{eq:appxbound}
    \sup_{\ma^{(k)}_{i_k}\in B(k),k\in[m]}&|f(\ma^{(1)}_{i_1},\ldots,\ma^{(m)}_{i_m})- P_{N,\ell-1}(\ma^{(1)}_{i_1},\ldots,\ma_{i_m}^{(m)})|\nonumber\\
    &\leq \sup_{1\leq i\leq N_{\text{tot}}}\sup_{\ma^{(k)}_{i_k}\in B(k),k\in[m]}|f(\ma^{(1)}_{i_1},\ldots,\ma^{(m)}_{i_m})-P_{I_i,\ell-1}(\ma^{(1)}_{i_1},\ldots,\ma_{i_m}^{(m)})|
    \nonumber\\
    &\leq \sum_{\bm \alpha\colon|\bm \alpha| = \ell}\frac{1}{\bm \alpha!}\sup_{\mx_k\in B(k),k\in[m]}\left|{\partial^{|\bm \alpha|}f(\mx_1,\ldots,\mx_m)\over \partial(\mx_1,\ldots,\mx_m)^{\bm \alpha}}\right|\left(1\over N\right)^\ell\nonumber\\
    &\leq \sum_{\bm \alpha\colon|\bm \alpha| = \ell}\left( M\over N\right)^\ell\nonumber\\&= \left(M\over N\right)^\ell {\ell+\sum_{k=1}^m s_k \choose\ell}.
\end{align}
Here we denote the unit ball $B(k) := B(\mathbb{R}^{s_k},\|\cdot\|_\infty)$ for simplicity.

Now, we define the approximation tensor $\tB_{N,\ell-1}\in\tP(\md,\ms,M)$ with entries
\begin{align}\label{eq:B}
    \tB_{N,\ell-1}(i_1,\ldots,i_m) = P_{N,\ell-1}\left(\ma_{i_1}^{(1)},\ldots,\ma_{i_m}^{(m)}\right), \ \text{for all}\ (i_1,\ldots,i_m)\in[d_1]\times\cdots \times [d_m]. 
\end{align}
By construction, the approximation error obeys 
\begin{align}
    \|\Theta-\tB_{N,\ell-1}\|_\infty \leq  \left(M\over N\right)^\ell {\ell+\sum_{k=1}^m s_k \choose\ell} \leq cM^{\ell}N^{-\ell} \ell^{m\bar s},
\end{align}
where $c=c(\ms, M)$ denotes a constant. 

\paragraph*{Low-rankness }
Now we compute the Tucker rank of $\tB_{N,\ell-1}$. It suffices to show 
\[
\text{rank}\left(\textup{Unfold}_1(\tB_{N,\ell-1})\right)\leq  (2N)^{\sum_{k=1}^m s_k} {\ell+s_1-2\choose s_1-1}\lesssim N^{m\bar s}(\ell+s_1)^{s_1}.
\] 
Given latent variables $\mx_k=(x_{k1},\ldots,x_{ks_k})\in\mathbb{R}^{s_k}$, we define the basis vector consisting of all monomials with degree up to $(\ell-1)$ by 
\begin{align}
&\phi(\mx_1,\ldots,\mx_m) \\
= &\ \{(\mx_1,\ldots,\mx_m)^{\bm \alpha}\colon |\bm \alpha|\leq \ell-1\}\\
 = &\ (1,x_{11},\ldots,x_{1s_1},x_{21},\ldots,x_{2s_2},\ldots,x_{11}x_{12},\ldots,x_{11}^{\ell-2}x_{21},\ldots,x_{ms_1}^{\ell-1},\ldots,x_{ms_m}^{\ell-1})^T \\
\in  &\  \mathbb{R}^{\ell+\sum_{k=1}^ms_k-2\choose \ell-1},
\end{align}
where the second line uses the convention~\eqref{eq:notation}. Here the dimension of $\phi$ is the number of all possible combinations of  $\{a_{ki_k}\}_{k\in[m],i_k\in[s_k]}$ with degree up to $(\ell-1)$. Then, the piecewise polynomial~\eqref{eq:piecewise} is expressed by 
\begin{align}
    P_{N,\ell-1}(\mx_1,\ldots,\mx_m) = \sum_{i=1}^{N_{\text{tot}}}\langle \phi(\mx_1,\ldots,\mx_m), \bm \beta_i \rangle \mathds{1}\{(\mx_1,\ldots,\mx_m)\in I_i\},
\end{align}
where $ \bm \beta_i \in\mathbb{R}^{\ell+\sum_{k=1}^ms_k-2\choose \ell-1}$ is a coefficient vector in the region $I_i.$

We now seek the expression of $\text{Unfold}_{1}(\tB_{N,\ell-1})$. We write $I_i=I_{i,1}\times I_{i,-1}$, where $I_{i,1}\subset\mathbb{R}^{s_1}$ and $I_{i,-1}\subset\mathbb{R}^{s_2\times \cdots \times s_m}$. Then
\[
\mathds{1}\left\{(\mx_1,\ldots,\mx_m)\in I_i\right\} = \mathds{1}\{\mx \in I_{i,1}\}\mathds{1}\{(\mx_2,\ldots,\mx_m)\in I_{i,-1}\}
\]
By the definition of $\tB_{N,\ell-1}$ in~\eqref{eq:B} and the relationship between tensor and its unfolding, we have
\begin{align}\label{eq:mat2}
&\text{Unfold}_1(\tB_{N,\ell-1}) \notag \\
=&\  \sum_{1\leq i\leq N_{\text{tot}}}\underbrace{\begin{pmatrix} \left\langle\phi(\ma_1^{(1)},\ma_1^{(2)}\ldots,\ma_1^{(m)}),\bm \beta_{i}\right\rangle, &\ldots,& \left\langle\phi(\ma_1^{(1)},\ma_{d_2}^{(2)}\ldots,\ma_{d_m}^{(m)}),\bm \beta_{i}\right\rangle\\\vdots&
\vdots&\vdots\\\left\langle\phi(\ma_{d_1}^{(1)},\ma_1^{(2)}\ldots,\ma_1^{(m)}),\bm \beta_{i}\right\rangle, &\ldots,&\left\langle\phi(\ma_{d_1}^{(1)},\ma_{d_2}^{(2)}\ldots,\ma_{d_m}^{(m)}),\bm \beta_{i}\right\rangle\end{pmatrix}}_{=: \mM_i} \circ \\
& \KeepStyleUnderBrace{\begin{pmatrix}
\mathds{1}\{\ma_{1}^{(1)} \in  I_{i,1}\}\mathds{1}\{(\ma^{(2)}_1,\ldots,\ma^{(m)}_1) \in  I_{i,-1}\}, & \ldots,& \mathds{1}\{\ma_{1}^{(1)} \in  I_{i,1}\}\mathds{1}\{(\ma^{(2)}_{d_2},\ldots,\ma^{(m)}_{d_m}) \in  I_{i,-1}\} \\
\vdots & \vdots & \vdots \\
\mathds{1}\{\ma_{d_1}^{(1)} \in  I_{i,1}\}\mathds{1}\{(\ma^{(2)}_1,\ldots,\ma^{(m)}_1) \in  I_{i,-1}\}, & \ldots,& \mathds{1}\{\ma_{d_1}^{(1)} \in  I_{i,1}\}\mathds{1}\{(\ma^{(2)}_{d_2},\ldots,\ma^{(m)}_{d_m}) \in  I_{i,-1}\} \\
 \end{pmatrix}}_{=: \mathds{1}_i},
 \end{align}
 where $\mM_i$ and $\mathds{1}_i$ are matrices of size $d_1$-by-$\prod_{k=2}^m d_k$, and $\circ$ denotes entrywise matrix product. We claim that
 \begin{equation}\label{eq:rank}
 \text{rank}(\mathds{1}_i)=1 \quad \text{and}\quad \text{rank}(\mM_i) \leq {{\ell+s_1-2}\choose{s_1-1}} \  \text{for all}\ i\in[N_{\text{total}}].
 \end{equation}
The property~\eqref{eq:rank} will be provided in the end of proof. Combining~\eqref{eq:mat2} and~\eqref{eq:rank} yields the desired conclusion,
\[
   \text{rank(Unfold}_1(\tB_{N,\ell-1}))\leq N_{\text{tot}} \text{rank}(\mM_i) \leq (2N)^{\sum_{k=1}^m s_k} {\ell+s_1-2\choose s_1-1}\lesssim N^{m\bar s}(\ell+s_1)^{s_1}.
\]
 
Finally, it suffices to verify~\eqref{eq:rank}. The fact $\text{rank}(\mathds{1}_i)=1$ comes from the factorization
 \begin{align}
\mathds{1}_i = \begin{pmatrix}
\mathds{1}\{\ma_{1}^{(1)} \in  I_{i,1}\}\\
\vdots\\
 \mathds{1}\{\ma_{d_1}^{(1)} \in  I_{i,1}\}
 \end{pmatrix}
  \begin{pmatrix}
\mathds{1}\{(\ma^{(2)}_1,\ldots,\ma^{(m)}_1) \in  I_{i,-1}\}, & \ldots,&\mathds{1}\{(\ma^{(2)}_{d_2},\ldots,\ma^{(m)}_{d_m}) \in  I_{i,-1}\}
 \end{pmatrix}.
\end{align}
We now show $\text{rank}(\mM_i) \leq {{\ell+s_1-2}\choose{s_1-1}}$. Notice that
\[
\mM_i=\begin{pmatrix}
\begin{pmatrix}
\phi(\ma_1^{(1)},\ma_1^{(2)}\ldots,\ma_1^{(m)})^T\\\vdots\\\phi(\ma_{d_1}^{(1)},\ma_1^{(2)}\ldots,\ma_1^{(m)})^T
\end{pmatrix}\bm \beta_{i}, &\cdots, &\begin{pmatrix}
\phi(\ma_1^{(1)},\ma_{d_2}^{(2)}\ldots,\ma_{d_m}^{(m)})^T\\\vdots\\\phi(\ma_{d_1}^{(1)},\ma_{d_2}^{(2)}\ldots,\ma_{d_m}^{(m)})^T
\end{pmatrix}\bm \beta_{i}
\end{pmatrix}.
\]
For any given $(i_2,\ldots,i_m)\in[d_2]\times\cdots\times [d_m]$, the column of $\mM$ obeys
\begin{align}
 &   \begin{pmatrix}
\phi(\ma_1^{(1)},\ma_{i_2}^{(2)}\ldots,\ma_{i_m}^{(m)})^T
\\\phi(\ma_2^{(1)},\ma_{i_2}^{(2)}\ldots,\ma_{i_m}^{(m)})^T\\\vdots\\\phi(\ma_{d_1}^{(1)},\ma_{i_2}^{(2)}\ldots,\ma_{i_m}^{(m)})^T
\end{pmatrix}\bm \beta_{i}\\
\in &\ \text{Column space of}\underbrace{\left\{\begin{pmatrix} 1\\1\\\vdots\\1\end{pmatrix},\begin{pmatrix} a^{(1)}_{11}\\a^{(1)}_{21}\\\vdots\\a^{(1)}_{d_11}\end{pmatrix},\begin{pmatrix} a^{(1)}_{12}\\a^{(1)}_{22}\\\vdots\\a^{(1)}_{d_12}\end{pmatrix},\ldots,\begin{pmatrix} a^{(1)}_{1s_1}\\a^{(1)}_{2s_1}\\\vdots\\a^{(1)}_{d_1s_1}\end{pmatrix},\ldots,\begin{pmatrix} (a^{(1)}_{1s_1})^{(\ell-1)}\\(a^{(1)}_{2s_1})^{(\ell-1)}\\\vdots\\(a^{(1)}_{d_1s_1})^{(\ell-1)}\end{pmatrix}\right\}}_{=:\text{Basis}},
\end{align}
where $\text{Basis}$ is a matrix with $j$-th row consisting of all monomials from the collection
\[
\{ (a^{(1)}_{j1},a^{(1)}_{j2},\ldots,a^{(1)}_{js_1})^{\bm \alpha}\colon |\bm \alpha|\leq \ell-1\}.
\]
By counting the number of monomials with degree up to $(\ell-1)$ from a collection of $s_1$ elements, we have
\[
\text{rank}(\mM_i)\leq \text{rank}(\text{Basis})\leq \text{number of rows}(\text{Basis}) \leq {\ell + s_1 -2 \choose s_1-1}.
\]
Therefore, the property~\eqref{eq:rank} is verified. 
\end{proof}

\subsection{Proof of Theorem~\ref{thm:lseupper}}
\begin{proof}
By Theorem~\ref{thm:rank}, for any $\Theta\in \tP(\md,\ms,M)$, we can always find a tensor $\tX_r$ whose Tucker rank is 
$(r,\ldots,r)$ satisfying
\begin{align}\label{eq:f0}
    \FnormSize{}{\Theta - \tX_r}^2\leq d_*e^{-c(\ms, M)r^{1/\bar s}}. 
\end{align}

By triangular inequality, 
\begin{align}\label{eq:full}
    \FnormSize{}{\hat\Theta^{\text{LSE}}-\Theta}\leq \FnormSize{}{\hat\Theta^{\text{LSE}}-\tX_r} + \FnormSize{}{\tX_r - \Theta}.
\end{align}
Since the second term in the above equation is well bounded by \eqref{eq:f0}, it suffices to bound the first term.
Because $\hat\Theta^{\text{LSE}}$ is a global optimizer of squared loss, by Taylor expansion,
we have
\begin{align}\label{eq:f}
    \FnormSize{}{\hat\Theta^{\text{LSE}}-\tX_r}&\leq \left\langle \frac{\hat\Theta^{\text{LSE}}-\tX_r}{\FnormSize{}{\hat\Theta^{\text{LSE}}-\tX_r}},\tE + (\Theta-\tX_r)\right\rangle\nonumber\\&\leq
    \sup_{\text{Tucker-rank}(\Theta_1),\text{Tucker-rank}(\Theta_2)\leq (r,\ldots ,r)}\left\langle \frac{\Theta_1-\Theta_2}{\FnormSize{}{\Theta_1-\Theta_2}},\tE + (\Theta-\tX_r)\right\rangle\nonumber\\&\leq
    \sup_{\text{Tucker-rank}(\Theta_3)\leq 2(r,\ldots,r),\FnormSize{}{\Theta_3}\leq 1}\left\langle \Theta_3,\tE\right\rangle + \FnormSize{}{\Theta-\tX_r},
\end{align}
where the last inequality is from the fact that $\text{Tucker-rank}(\Theta_1-\Theta_2)\leq \text{Tucker-rank}(\Theta_1)+\text{Tucker-rank}(\Theta_2)\leq 2(r,\ldots,r)$ 
Therefore, applying Lemma~\ref{lem:0} and \eqref{eq:f0} to \eqref{eq:f} gives us 
\begin{align}
     \FnormSize{}{\hat\Theta^{\text{LSE}}-\tX_r}^2\lesssim r^m + \bar d r + d_* e^{-c(\ms, M) r^{1/\bar s}},
\end{align}
with high probability $1-c\exp(-r\underline d)$.
Setting $r = c(\ms, M)^{-\bar s} \log^{\bar s} \left(\frac{d_*}{\underline d}\right)$ completes the proof.
\end{proof}

\subsection{Proof of Theorem~\ref{thm:lselower}}
We first introduce Lemma~\ref{lem:1} and then prove Theorem~\ref{thm:lselower}.
\begin{lem}[Lemma 8 in \citet{wang2018learning}]\label{lem:1}
Define the class of Tucker low-rank tensors by
\begin{equation}\label{eq:tucker}
\tT(\md,\mr,M) = \{\Theta\in\mathbb{R}^{d_1\times \cdots\times d_m}\colon \textup{Tucker-rank}(\Theta)\leq (r_1,\ldots,r_m),\|\Theta\|_{\infty}\leq M\}.
\end{equation}
For any given constant $0\leq\gamma\leq 1$, there exists a finite set of tensors $\tX =\{\Theta_i\colon i = 1,\ldots\}\subset \tT(\md,\mr,M)$  satisfying the following four properties:
\begin{enumerate}[label=(\roman*)]
    \item $\text{Card}(\tX)\geq 2^{\bar r \bar d/8}$ + 1, where $\text{Card}(\cdot)$ denotes the cardinality of the set. 
    \item $\tX$ contains the zero tensor $\bm 0\in\mathbb{R}^{d_1\times \cdots\times d_m}$;
    \item $\|\Theta\|_\infty \leq \gamma\left(M\wedge \sqrt{\bar r\bar d\over d_*}\right)$ for all elements $\Theta\in\tX$;
    \item $\FnormSize{}{\Theta_i-\Theta_j}\geq {\gamma \over 4}\left( M\sqrt{d_*} \wedge \sqrt{\bar r\bar d}\right)$ for any two distinct elements $\Theta_i\neq\Theta_j\in\tX.$
\end{enumerate}
\end{lem}

\begin{rmk}[Minimax lower bound for low-rank tensor estimation]\label{rmk:tsybakov}
Consider the Gaussian observation model
\[
\tY=\Theta+\tE, 
\]
where $\Theta \in \tT(\md,\mr,M)$ is the low-rank signal tensor of interest, and $\tE$ is the noise tensor consisting of i.i.d.\ standard Gaussian random variables. A direct application of \citet[Theorem 2.5]{tsybakov2009introduction} and Lemma~\ref{lem:1} to the above setting implies that
\begin{align}\label{eq:tsybakov}
   \inf_{\hat\Theta}\sup_{\Theta\in\tT(\md,\mr,M)}\mathbb{P}\left(\FnormSize{}{\hat\Theta-\Theta}\gtrsim \sqrt{\bar r\bar d}\right)\geq   \inf_{\hat\Theta}\sup_{\Theta\in\tX}\mathbb{P}\left(\FnormSize{}{\hat\Theta-\Theta}\gtrsim \sqrt{\bar r\bar d}\right)\geq p_0,
\end{align}
for some universal constant $p_0>0.$ 
\end{rmk}

\begin{proof}[Proof of Theorem~\ref{thm:lselower}]
For any given core tensor $\tC\in\mathbb{R}^{s_1\times \cdots\times s_m}$, we define an analytic function $f_{\tC}\colon\mathbb{R}^{s_1}\times \cdots\times \mathbb{R}^{s_m}\rightarrow \mathbb{R}$ by 
\begin{align}
f_{\tC}(\ma_1,\ldots,\ma_m) = \tC\times_1\ma_1^T\times_2 \cdots\times_m \ma_m^T \quad\text{ for any } \ma_k\in\mathbb{R}^{s_k},k\in[m].
\end{align}
Based on the notion of $f_{\tC}$,  we can rewrite the class of low-rank tensors in~\eqref{eq:tucker} by
\begin{align}
\tT(\md,\ms,M) &= \bigg\{\Theta\in\mathbb{R}^{d_1\times \cdots\times d_m}\colon \Theta(i_1,\ldots,i_m) = f_{\tC}(\ma_{i_1}^{(1)},\ldots,\ma_{i_m}^{(m)}) \text{ for some } \|\tC\|_\infty \leq M, \\&\quad \ma_{i_k}^{(k)}\in B(\mathbb{R}^{s_k},\|\cdot\|_\infty),\text{ and all } i_k\in[d_k], k\in[m].\bigg\}.
\end{align}
Notice that  $\tT(\md,\ms,M)\subset \tP(\md,\ms,M)$ by definition; see Example~\ref{ex:tucker} in the main paper.
A direct application of Remark~\ref{rmk:tsybakov} completes the proof, because
\begin{align*}
\inf_{\hat\Theta}\sup_{\Theta\in\tP(\md,\ms,M)}\mathbb{P}\left(\FnormSize{}{\hat\Theta-\Theta}\gtrsim \sqrt{\bar s \bar d}\right)
&\geq \inf_{\hat\Theta}\sup_{\Theta\in\tT(\md,\ms,M)}\mathbb{P}\left(\FnormSize{}{\hat\Theta-\Theta}\gtrsim \sqrt{\bar s \bar d}\right) \geq p_0,
\end{align*}
for some universal constant $p_0>0.$
\end{proof}

\subsection{Proof of Theorem~\ref{thm:polylower}}
\begin{proof}

The proof of Theorem~\ref{thm:polylower} leverages the result in detecting a constant planted structure in higher-order tensors~\citep{luo2022tensor}. 
Assume that the true signal $\Theta\in\mathbb{R}^{d_1\times \cdots \times d_m}$ has constant high-order clustering (CHC) structure with a given $\mr=(r_1,\ldots,r_m)$ such that 
\begin{align}
    \Theta \in \Theta_{\text{CHC}}(\md,\mr, \lambda) := \{ \lambda \mathds{1}_{I_1}\circ\cdots\circ\mathds{1}_{I_m}\colon  \text{Card}(I_k) = r_k \text{ for all }k\in[m]\},
\end{align}
where $I_i\subset [d_i]$ is the subset of indices, $\text{Card}(\cdot)$ denotes the cardinality of the set, $\mathds{1}_{I_i}$ is the $d_i$-dimensional indicator vector such that $(\mathds{1}_{I_i})_j=1$ if $j\in I_i$ and 0 otherwise. Note that $\Theta_{\text{CHC}}$ consists of CP rank-1 tensors with 
$\|\Theta\|_{\infty}= \lambda$. From Example~\ref{ex:cp} in the main paper with CP rank 1, we have that $\Theta_{\text{CHC}}(\md,\mr,\lambda) \subset \tP(\md,\ms,\lambda)$ for all $\ms$ and $\lambda$.

We consider the hypothesis test of detecting CHC based on the observed tensor $\tY$,
\begin{align}\label{eq:chct}
    H_0\colon \Theta = 0\quad\text{v.s.}\quad H_1\colon \Theta\in\Theta_{\text{CHC}}(\md,\mr, \lambda).
\end{align}
The following proposition provides the asymptotic regime for impossible detection of CHC with computationally feasible test $\phi$ under Conjecture~\ref{conj:1}.

\begin{prop}[Theorem 15 in \citep{luo2022tensor}]\label{prop:CHC} 
Consider the CHC detection problem \eqref{eq:chct} in the Gaussian observation model~\eqref{eq:gmodel} under the asymptotic regime $d\rightarrow\infty$ satisfying
\begin{align}\label{eq:equal}
  d =  d_1= \cdots= d_m ,\quad
    r = r_1=\cdots=r_m = d^\alpha,\quad \lambda  = d^{-\beta},
\end{align}with $0\leq\alpha\leq1$ and $\beta >(m\alpha-m/2)\vee 0.$ Then, under the HPC detection Conjecture~\ref{conj:1}, for all polynomial-time test sequence $\{\phi\}_d\colon \tY\mapsto \{0,1\}$, we have
\begin{align}
    \liminf_{d\rightarrow\infty} \text{Err} (\phi)\geq\frac{1}{2},
\end{align}
where $\text{Err}(\phi)$ is defined in~\eqref{eq:err}. 
\end{prop}


For technical convenience, we consider the setting of equal dimension 
\[
d_1 = \cdots = d_m = d, \quad \text{and}\quad r_1 = \ldots = r_m = c_1d^{\alpha} \quad \text{for a constant $c_1>0$}. 
\]
Based on conditions in Conjecture~\ref{conj:1} and Proposition~\ref{prop:CHC} in the main paper, set $\alpha=1/2$ and $\lambda = c_2d^{-\epsilon/2}$ for a constant $c_2>0$ and any $\epsilon >0$, so that all polynomial-time tests $\phi$ satisfy $\liminf_{d\rightarrow\infty}\text{Err}(\phi)\geq 1/2.$  

We prove by contradiction. Now, assume that there exists a hypothetical estimator $\hat\Theta$ from a polynomial-time algorithm  that attains the estimation error rate $d_*^{1/2-\epsilon} = d^{m/2-\epsilon}$ (for notational simplicity, let $m$ absorbed in to $\epsilon$). Specifically, there exist constants $b>0$ and $\varepsilon>0$, such that
\begin{align}\label{eq:esterror}
   \limsup_{d\rightarrow\infty}\frac{1}{d^{m/2-\epsilon}} \sup_{\Theta\in\tP(\md,\ms, \lambda)}\mathbb{E}\FnormSize{}{\hat\Theta-\Theta}^2\leq b.
\end{align}
By Markov's inequality, the statement \eqref{eq:esterror} implies that, when $d$ is sufficiently large, then for all $\Theta\in\tP(\md,\ms,\lambda)$ and all $u \geq 0$, 
\begin{align}\label{eq:contrary}
    \FnormSize{}{\hat\Theta-\Theta}\leq u d^{{m\over 4}-{\epsilon\over 2}}, 
\end{align}
 with probability at least $1-b/u$. In particular, the statement~\eqref{eq:contrary} holds for all $\Theta\in\Theta_{\text{CHC}}(\md,\mr, \lambda)\subset \tP(\md,\ms,\lambda)$.

 Consider the hypothesis test in \eqref{eq:chct} in the main paper. Then we employ the following test 
\begin{align}
    \phi(\tY) = \mathds{1}(\FnormSize{}{\hat\Theta}\geq ud^{{m\over 4}-{\epsilon\over 2}}).
\end{align}
The Type I error of the test $\phi$ is controlled by
\begin{align}
    \mathbb{P}_{H_0}(\FnormSize{}{\hat\Theta}\geq ud^{{m\over 4}-{\epsilon\over 2}})  = \mathbb{P}_{H_0}(\FnormSize{}{\hat\Theta-\Theta}\geq ud^{{m\over 4}-{\epsilon\over 2}}) \leq b/u.
\end{align}
For Type II error, we obtain
\begin{align}
\sup_{\Theta\in\Theta_{\text{CHC}}(\md,\mr, \lambda)}\mathbb{P}_\Theta(\phi(\tY) = 0) &= \sup_{\Theta\in\Theta_{\text{CHC}}(\md,\mr, \lambda)}\mathbb{P}_\Theta(\FnormSize{}{\hat\Theta}< u d^{{m\over 4}-{\epsilon\over 2}})\\&\leq\sup_{\Theta\in\Theta_{\text{CHC}}(\md,\mr,\lambda)}\mathbb{P}_\Theta(\FnormSize{}{\hat\Theta-\Theta}^2> \FnormSize{}{\Theta}^2-u^2d^{{m\over 2}-\epsilon})\\&\stackrel{(*)}{\leq}\sup_{\Theta\in\Theta_{\text{CHC}}(\md,\mr, \lambda)}\mathbb{P}_\Theta(\FnormSize{}{\hat\Theta-\Theta}^2> u^2d^{{m\over 2}-\epsilon})\\
&\stackrel{(**)}{\leq} b/u.
\end{align}
The inequality $(*)$ holds because 
\begin{align}\label{eq:regime}
    \FnormSize{}{\Theta}^2\geq \lambda^2 r^m = c_1^mc_2^2d^{\frac{m}{2}-\epsilon} \geq 2 u^2 d^{{m\over 2}-\epsilon},
\end{align}
where the last inequality is true under the constraint $c_1^mc_2>2u^2$. We can always choose such constants $c_1$ and $c_2$  given some value $u$. The inequality $(**)$ holds because of the statement~\eqref{eq:contrary}. 
Putting Type I and II errors together, we obtain
\begin{align}\label{eq:wrong}
\text{Err}(\phi) &\leq \mathbb{P}_{\tH_0}(\phi(\tY) = 1)+ \sup_{\Theta\in\Theta_{\text{CHC}}(\md,\mr, \lambda)}\mathbb{P}_{\Theta}(\phi(\tY) =0)\\&\leq 2b/u < 1/2,
\end{align}
for $u>4b$. The statement~\eqref{eq:wrong} contradicts Proposition~\ref{prop:CHC} in the main paper. Therefore, there is no polynomial-time $\hat\Theta$ satisfying \eqref{eq:esterror}.
\end{proof}

\subsection{Proof of Theorem~\ref{thm:polyupper}}
\begin{proof}
We set $\varepsilon=\sqrt{\bar d}$ in the proof of Proposition~\ref{cor:rank}. Let $\tX$ be the approximated tensor satisfying 
\begin{align}\label{eq:fbound2}
    \FnormSize{}{\Theta-\tX}\lesssim \sqrt{\bar d},
\end{align}
with $\textup{Tucker-rank}(\tX) = (r_1,\ldots,r_k)$ where $r_k =  c(\ms, M)^{-\bar s}\log^{\bar s}\left(d_*/\bar d \right)$ for all $k\in[m].$
Notice that $\tX$ also satisfies the following inequality from \eqref{eq:fbound2}
\begin{align}\label{eq:sbound}
    \SnormSize{}{\Theta - \tX}:= \max_{k\in[m]}\SnormSize{}{\textup{Unfold}_k(\Theta)-\textup{Unfold}_k(\tX)}\lesssim \sqrt{\underline d}.
\end{align}
Define $\tX_{\perp}:= \Theta - \tX$. For each $k = 1,\ldots,m$, we denote
 \[X_k = \textup{Unfold}_k(\tX ),\quad X_{k,\perp} = \textup{Unfold}_k(\tX_{\perp}),\quad  E_k = \textup{Unfold}_k(\tE),\quad Y_k =\textup{Unfold}_k(\tY),\]
and define $Z_k = X_{k,\perp}+E_k.$
Now we review the notation in the double-projection spectral algorithm in the main paper 
\begin{align}\label{eq:est}
    \tilde U_k &= \text{SVD}_{r_k}(Y_k) \nonumber\\
    \hat U_k &=\text{SVD}_{r_k}\left(\textup{Unfold}_k\left(\tY\times_1\tilde U_1^T\times \cdots \times_{k-1}\tilde U_{k-1}^T\times_k\tilde U_{k+1}^T\times \cdots \times_m \tilde U_m^T\right)\right)\nonumber\\
    \hat\Theta &= \tY\times_1(\hat U_1\hat U_1^T)\times \cdots\times_m (\hat U_m\hat U_m^t).
\end{align}
 Denote $U_k =\text{SVD}_{r_k}
(X_k)$. For some
constant $C_0>0$ which will be specified later, define 
\begin{align}
    r_k' = \max\{r'\in\{0,\ldots,d_k\}\colon \sigma_{r'}(X_k)\geq C_0(d_*^{1/4}\vee \bar d^{1/2})\}.
\end{align}
We set $r'_k = 0$ if
 $\sigma_1(X_k) < C_0(d_*^{1/4}\vee \bar d^{1/2})$.
We use $U_k'$ to denote the leading $r_k'$
 singular vectors of
$U_k$ and use $V_k$
to denote the rest $r_k- r'_k$
singular vectors and thus $U_k$ can be written as $[U_k',V_k]$.
We next define
\begin{align}
    X_k' = X_k\left(\mathbb{P}_{U_{k+1}'}\otimes\cdots\otimes \mathbb{P}_{U_m'}\otimes\cdots\otimes\mathbb{P}_{U_{k-1}'}\right),
\end{align}
where $\mathbb{P}_{U} = UU^T$ for any orthonormal matrix $U\in\mathbb{R}^{d\times r}$. We also denote
\begin{align}
    \bar X_k &= X_k\left(\tilde U_{k+1}\otimes \cdots\otimes \tilde U_m\otimes \cdots\otimes \tilde U_{k-1}\right),\\
    \bar X_{k,\perp} &= X_{k,\perp}\left(\tilde U_{k+1}\otimes \cdots\otimes \tilde U_m\otimes \cdots\otimes \tilde U_{k-1}\right),\\
    \bar Y_k &= Y_k\left(\tilde U_{k+1}\otimes \cdots\otimes \tilde U_m\otimes \cdots\otimes \tilde U_{k-1}\right),\\
    \bar E_k &= E_k\left(\tilde U_{k+1}\otimes \cdots\otimes \tilde U_m\otimes \cdots\otimes \tilde U_{k-1}\right)\\
    \bar Z_k &=  \bar X_{k,\perp}+\bar E_k.
\end{align}
Now we bound 
\begin{align}\label{eq:start12}
\FnormSize{}{\hat\Theta-\Theta}\leq    \underbrace{\FnormSize{}{\Theta\times_1(\hat U_1\hat U_1^T)\times\cdots \times_m(\hat U_m\hat U_m^T)-\Theta}}_{(*)}+\underbrace{\FnormSize{}{\tE\times_1\hat U_1^T\times\cdots\times_m\hat U_m^T}}_{(**)}.
\end{align}
To bound $(*)$, we have
\begin{align}\label{eq:star}
    (*)&\leq \sum_{k\in[m]}\FnormSize{}{(I-\hat U_k\hat U_k^T) \Theta_k}\nonumber\\&\leq\sum_{k\in[m]}\left(\FnormSize{}{(I-\hat U_k\hat U_k^T) X_k}+\FnormSize{}{(I-\hat U_k\hat U_k^T) X_{k,\perp}}\right)\nonumber\\& \leq\sum_{k\in[m]}\left(\FnormSize{}{\hat U_{k,\perp}^T X_k}+\FnormSize{}{X_{k,\perp}}\right)\nonumber\\&\leq \sum_{k\in[m]}\left(\FnormSize{}{\hat U_{k,\perp}^T X_k'}+\FnormSize{}{X_k-X_k'}+\FnormSize{}{X_{k,\perp}}\right).
\end{align}
Therefore, it suffices to bound  $\FnormSize{}{X_k-X_k'}$ and $\FnormSize{}{\hat U_{k,\perp}^T X_k'}$ because $\FnormSize{}{X_{k,\perp}}\leq\sqrt{\bar d}$ by \eqref{eq:fbound2}.
\begin{enumerate}
    \item Bound of $\FnormSize{}{X_k-X_k'}$: For notation simplicity, we focus on $k = 1$, while the analysis
for other modes can be similarly carried on. We have
\begin{align}\label{eq:comb1}
    \FnormSize{}{X_1-X_1'}&\leq\FnormSize{}{X_1\left((\mathbb{P}_{U_2'}+\mathbb{P}_{V_2'})\otimes\cdots\otimes(\mathbb{P}_{U_m'}+\mathbb{P}_{U_m'})-\mathbb{P}_{U_2'}\otimes\cdots\otimes \mathbb{P}_{U_m'}\right)}\nonumber\\&\leq \sum_{k=2}^m\FnormSize{}{V_k'^T X_k}\nonumber\\
    &\leq\sum_{k=2}^m\sqrt{r_k-r_k'}\sigma_{r_k'+1}(X_k)\nonumber\\&\leq \sum_{k=2}^m  C_0\sqrt{r_k}(d_*^{1/4}+\bar d^{1/2}).
\end{align}

\item Bound of $\FnormSize{}{\hat U_{k,\perp}^T X_k'}$: We have the following two inequalities
\begin{align}\label{eq:1}
    \FnormSize{}{\hat U_{k,\perp}^TX_k'(\tilde U_{k+1}\otimes \cdots \otimes \tilde U_m\otimes \tilde U_1\otimes \cdots \otimes \tilde U_{k-1})}\leq\FnormSize{}{\hat U_{k,\perp}^T\bar X_k}+\FnormSize{}{X_k-X_k'}
\end{align}
\begin{align}\label{eq:2}
      \FnormSize{}{\hat U_{k,\perp}^TX_k'&(\tilde U_{k+1}\otimes \cdots \otimes \tilde U_m\otimes \tilde U_1\otimes \cdots \otimes \tilde U_{k-1})}\nonumber\\&= \FnormSize{}{\hat U_{k,\perp}^TX_k'(\mathbb{P}_{U_{k+1}'}\tilde U_{k+1}\otimes \cdots \otimes \mathbb{P}_{U_{m}'}\tilde U_m\otimes \mathbb{P}_{U_{1}'}\tilde U_1\otimes \cdots \otimes\mathbb{P}_{U_{k-1}'} \tilde U_{k-1})}\nonumber\\&\geq \FnormSize{}{\hat U_{k,\perp}^T X_k'}\prod_{\ell\neq k}\sigma_{r_\ell'}(U_\ell^T\tilde U_\ell)\nonumber\\&= \FnormSize{}{\hat U_{k,\perp}^T X_k'}\prod_{\ell\neq k}\sqrt{1-\SnormSize{}{\tilde U_{\ell,\perp}^TU_\ell'}^2}.
\end{align}
Combining \eqref{eq:1} and \eqref{eq:2} yields
\begin{align}\label{eq:comb}
    \FnormSize{}{\hat U_{k,\perp}^T X_k'}\prod_{\ell\neq k}\sqrt{1-\SnormSize{}{\tilde U_{\ell,\perp}^TU_\ell'}^2}\leq\FnormSize{}{\hat U_{k,\perp}^T\bar X_k}+\FnormSize{}{X_k-X_k'}.
\end{align}
Now, we bound $\FnormSize{}{\hat U_{k,\perp}^T\bar X_k}$ and $\SnormSize{}{\tilde U_{\ell,\perp}^TU_\ell'}$ to obtain the upper bound of $\FnormSize{}{\hat U_{k,\perp}^T X_k'}$.

The upper bound for $ \FnormSize{}{\hat U_{k,\perp}^T\bar X_k}$ follows from Lemma~\ref{lem:projec} and the fact that  $\bar Y_k = \bar X_k+ \bar Z_k$ as
\begin{align}\label{eq:ux}
     \FnormSize{}{\hat U_{k,\perp}^T\bar X_k}&\leq  2\sqrt{r_k}\SnormSize{}{\bar Z}\nonumber\\&\leq2\sqrt{r_k}\left(\SnormSize{}{\bar X_{k,\perp}}+\SnormSize{}{\bar E_k}\right)\nonumber\\&\leq
     2\sqrt{r_k}\left(\SnormSize{}{ X_{k,\perp}}+\SnormSize{}{\bar E_k}\right)\nonumber\\&\lesssim\sqrt{r_k\bar d}+\sqrt{r_* }+\sum_{\ell\in[m]}\sqrt{ r_{\ell}\bar rd_\ell},
\end{align}
where the last line uses $\SnormSize{}{\tX_{\perp}}\leq \sqrt{\bar d}$ from  \eqref{eq:sbound}, the definition of $\bar E_k$ and Lemma~\ref{lem:subg}.

By Lemma~\ref{lem:uu}, we bound $\SnormSize{}{\tilde U_{k,\perp}^TU_k'}$ with probability at least $1 -C\exp(-c\underline d)$, for each $k \in[m]$,
\begin{align}\label{eq:3}
    \SnormSize{}{\tilde U_{k,\perp}^TU_k'}&\leq C\left(\frac{\sqrt{d_k}+\SnormSize{}{X_{k,\perp}}}{\sigma_{r'}(X)}+\frac{\sqrt{d_*}+\sqrt{d_*/d_k}\SnormSize{}{X_{k,\perp}}+\SnormSize{}{X_{k,\perp}}^2}{\sigma_{r'}^2(X)}\right)\nonumber\\&\leq\frac{C}{C_0}\left(\frac{\sqrt{d_k}+\sqrt{\bar d}}{\sqrt{\bar d}}+\frac{\sqrt{d_*}+\bar d}{\sqrt{d_*}\vee \bar d}\right)\nonumber\\&\leq\frac{1}{\sqrt{2}}
\end{align}  for sufficiently large $C_0\geq15$ where 15 is set to satisfy the condition of Lemma~\ref{lem:uu}.

Finally, plugging \eqref{eq:comb1}, \eqref{eq:ux}, and \eqref{eq:3} into \eqref{eq:comb} yields
\begin{align}\label{eq:comb2}
    \FnormSize{}{\hat U_{k,\perp}^T X_k'}&\leq 2^{\frac{m-1}{2}}\left(\FnormSize{}{\hat U_{k,\perp}^T\bar X_k}+\FnormSize{}{X_k-X_k'}\right)\\
    &\lesssim r_*^{1/2} + \bar r\bar d^{1/2}+ \bar r^{1/2}d_*^{1/4}.
\end{align}
\end{enumerate}
Applying \eqref{eq:comb1} and \eqref{eq:comb2} to \eqref{eq:star} proves
\begin{align}
    (*)\lesssim r_*^{1/2} + \bar r\bar d^{1/2}+ \bar r^{1/2}d_*^{1/4}.
\end{align}
Notice that $(**)$ term in~\eqref{eq:start12} is bounded by $C(\sqrt{r_*}+\sum_{\ell\in [m]}\sqrt{d_\ell r_\ell})$ by Lemma~\ref{lem:subg} with probability at least $1-\exp(-c\underline d)$.
Combining upper bound of $(*)$ and $(**)$, we finally obtain 
\begin{align}\label{eq:final}
    \FnormSize{}{\hat\Theta-\Theta}\lesssim r_*^{1/2} + \bar r\bar d^{1/2}+ \bar r^{1/2}d_*^{1/4}.
\end{align}
Plugging  $r_k =  c(\ms, M)^{-\bar s}\log^{\bar s}\left(d_*/\bar d\right)$ for all $k\in[m]$ into \eqref{eq:final} completes the proof.
\end{proof}

\section{Discussion}\label{sec:disc}
In this article, we propose the latent variable tensor model for the high rank tensor estimation problem. 
The latent variable tensor model provides a rigorous justification for the empirical success of low-rank methods despite the prevalence of high rank tensors in real data applications. We propose two estimation methods: statistically optimal but computationally impossible least-square estimation and computationally optimal DSE estimation. The intrinsic gap between statistical and computationally optimal rates is discovered. Numerical analysis demonstrates the effectiveness and applicability of our methods. 

There are several possible extensions of our work. We discuss the challenges and limitations. 
\noindent
{\bf Extension of noise models.} Our current theory assumes that the noise tensor consists of i.i.d. entries from Gaussian distribution. In fact, all theorems except Theorem~\ref{thm:polylower} can be extended to i.i.d. sub-Gaussian noise tensors.
Specifically, we can extend our results to the following subGaussian noise model:
\begin{itemize}[itemsep = 5pt]
    \item Centered: $\mathbb{E}(\tE(i_1,\ldots,i_m)) = 0$ for all $(i_1,\ldots,i_m)\in[d_1]\times \cdots\times[d_m]$.
    \item Sub-Gaussian: there exists a bounded $\sigma >0$ such that, for all $(i_1,\ldots,i_m)\in[d_1]\times \cdots\times[d_m]$,
    \begin{align}
    \mathbb{E}\exp(\lambda \tE(i_1,\ldots,i_m))\leq \exp^{\lambda^2\sigma^2/2},\quad \text{for all}\ \lambda \in\mathbb{R}.
    \end{align}
    \item Independent and identically distributed (i.i.d.): $\tE(i_1,\ldots,i_m)$'s are i.i.d. for all $(i_1,\ldots,i_m)\in[d_1]\times \cdots\times[d_m]$. 
\end{itemize}
In Section~\ref{sec:proof}, our proofs of Theorems~\ref{thm:lseupper}, \ref{thm:lselower}, and \ref{thm:polyupper} use only the properties of sub-Gaussianity, so extensions are straightforward. However, extending Theorem~\ref{thm:polylower} to subGaussian noise model remains challenging. The difficulty lies in the extension of Proposition~\ref{prop:CHC} to sub-Gaussianity. The main proof idea of Proposition~\ref{prop:CHC} is to construct the randomized polynomial-time algorithm $\varphi$, based on the average trick idea~\citep{ma2015computational} or the rejection kernel technique~\citep{brennan2018}, satisfying 
\begin{align}\label{eq:tv}
\text{TV}\left(\varphi(\text{HPC}(d,1/x,\tau)),\text{CHC}(\md,\mr,\lambda)\right)\rightarrow 0,
\end{align}
where we consider the total variation distance (TV)  between problems $\text{HPC}(d,1/x,\tau)$ and $\text{CHC}(\md,\mr,\lambda)$ under Gaussian noise model; see~\cite{luo2022tensor} for details. Unfortunately, the proof for \eqref{eq:tv} heavily uses explicit formula of standard normal distribution.  Due to the nature of the total variation distance,  bounding the total variation distance between HPC and CHC problems with an arbitrary sub-Gaussian distribution is challenging.  

In addition, our i.i.d.\ noise assumption precludes binary tensors, because Bernoulli entries are generally non-identically distributed.  Theorems~\ref{thm:rank}-\ref{thm:lselower} still hold for Bernoulli tensors, but Theorem~\ref{thm:polyupper} may fail. Our 
Theorem~\ref{thm:polyupper} presents the error rate of the DSE algorithm by using the perturbation bounds of singular space~\citep{han2022exact,cai2018rate}. 
Unfortunately, while perturbation bounds based on i.i.d.\ sub-Gaussian noises are well studied~\citep{cai2018rate,fan2018eigenvector}, those under independent but heteroskedastic sub-Gaussian noises are unknown.  Whether we can obtain similar results under heteroskedastic noises is an interesting extension.

\noindent
{\bf Random vs. fixed latent vectors.} Our current model assumes unknown but fixed latent variables in the signal tensor $\Theta$. One possible extension is to consider random latent variables in the generative model~\eqref{eq:LVM}. We can model the signal tensor by $\Theta(i_1,\ldots,i_m)=f(\ma^{(1)}_{i_1},\ldots,\ma_{i_m}^{(m)})$, with $\{\ma_{i_k}^{(k)}\colon i_k\in[d_k]\}$ are i.i.d.\ random vectors from some distributions. Similar random design models have been proposed for graphons and hypergraphons~\citep{chan2014consistent,gao2015rate,klopp2017oracle,balasubramanian2021nonparametric}. We now show our major theorems allow this random design under an extra bounded signal condition: $\|\Theta\|_\infty\leq \alpha$ for some constant $\alpha>0$. Specifically, we write $\ma_{i_k}^{(k)}\sim \tA$ i.i.d.\ for all $i_k\in[d_k]$, where $\tA$ denotes a distribution. We measure the estimation performance by the expected mean square error over all randomness defined as $\mathbb{E}_{\tY,\tA}(d_*^{-1}\FnormSize{}{\hat\Theta-\Theta}^2)$. In the proof of Theorem~\ref{thm:lseupper} and Theorem~\ref{thm:polyupper}, we have provided the upper bound of MSE by
\begin{align}\label{eq:expect}
  &  \mathbb{E}_{\tY}\left(d_*^{-1}\FnormSize{}{\hat\Theta-\Theta}^2\right) \notag \\
    =&\  \mathbb{E}_{\tY}\left(d_*^{-1}\FnormSize{}{\hat\Theta-\Theta}^2\mathds{1}_{\{\FnormSize{}{\hat\Theta-\Theta}^2\geq \epsilon^2\}}\right) + \mathbb{E}_{\tY}\left(d_*^{-1}\FnormSize{}{\hat\Theta-\Theta}^2\mathds{1}_{\{\FnormSize{}{\hat\Theta-\Theta}^2< \epsilon^2\}}\right)\nonumber \notag \\
    \leq &\  4\alpha^2\underbrace{\mathbb{P}_{\tY}(\FnormSize{}{\hat\Theta-\Theta}^2\geq \epsilon^2)}_{(*)} + \underbrace{d_*^{-1}\epsilon^2}_{(**)},
\end{align}
where the last inequality is from the fact that $\|\hat\Theta-\Theta\|_\infty\leq 2\alpha$. By setting $\epsilon^2 = \tilde\tO(\bar d)$ in Theorem~\ref{thm:lseupper} (or $\varepsilon^2=\tilde\tO(d_*^{1/2}+\bar d)$ in Theorem~\ref{thm:polyupper}), we make $(*)$ exponentially decreasing, thereby obtaining the upper bound of the expected mean square error as $(**)$. Now under the random design $\ma_{i_k}^{(k)}\sim \tA$ i.i.d.\ for all $i_k\in[d_k]$, the MSE in \eqref{eq:expect} becomes the conditional expectation given latent variables in $\tA$. Because the upper bound in \eqref{eq:expect} is uniform over $\tA$, we obtain the upper bound of the expected MSE over all randomness, by using
\[\mathbb{E}_{\tY,\tA}(d_*^{-1}\FnormSize{}{\hat\Theta-\Theta}^2) = \mathbb{E}_{\tA}\left(\mathbb{E}_{\tY}(d_*^{-1}\FnormSize{}{\hat\Theta-\Theta}^2\mid\tA)\right)\lesssim d_*^{-1}\epsilon^2.\]
Therefore, we reach the conclusion as in the fixed design. 

\noindent
{\bf Logarithmic factors in bounds.} Lastly, obtaining tighter lower bounds is of interest for the future work. 
We have shown that the statistical lower and upper bounds are $\tO(d^{-(m-1)})$ and $\tilde{\tO}(d^{-(m-1)})$, respectively. 
Similarly, the computational lower and upper bounds are $\tO(d^{m/2})$ and $\tilde \tO(d^{m/2})$, respectively. There is the logarithmic gap between  lower and upper bounds. We conjecture that lower bound analysis can be improved up to logarithm factors. We leave the possible improvement as future work.

\section*{Acknowledgements}
This research is supported in part by NSF CAREER DMS-2141865, DMS-1915978, DMS-2023239, EF-2133740, and funding from the Wisconsin Alumni Research foundation.

\bibliographystyle{Chicago.bst}
\bibliography{ref}  

\clearpage
\appendix
\section*{Appendix}
The appendix includes technical lemmas and extra simulation results.

\renewcommand{\thefigure}{S\arabic{figure}}
\setcounter{figure}{0}   
\renewcommand{\thetable}{S\arabic{table}}
\setcounter{table}{0}

\section{Technical Lemmas}
\begin{lem}[Lemma E.5 in \cite{han2022optimal}]\label{lem:0}
Assume all entries of $\tE\in\mathbb{R}^{d_1\times \cdots\times d_m}$ are independent mean zero  sub-Gaussian with proxy variance $\sigma$.  Then there exist some universal constants $C,c$ such that 
\begin{align}
    \sup_{\substack{\tT\in\mathbb{R}^{d_1\times \cdots \times d_m},\FnormSize{}{\tT}\leq 1\\ \text{rank}(\tT)\leq (r_1,\ldots,r_m)}}\langle \tT,\tE\rangle \leq C\sigma\left(r_* + \sum_{k=1}^m d_mr_k\right)^{1/2},
\end{align}
with probability at least $1-\exp(-c\sum_{k=1}^m d_kr_k).$
\end{lem}

\begin{lem}[Lemma 8 in \cite{han2022exact}]\label{lem:subg}
Let $E \in \mathbb{R}^{d_1\times\cdots\times d_m}$
 be a noise tensor whose each entry has independent mean-zero sub-Gaussian
distribution with $\sigma = 1$ without loss of generality. Fix $U^*_k \in\mathbb{O}_{d_k,r_k}$. Then with probability at least $1-\exp(-c\underline d)$, the
following holds:
\begin{align}
    \SnormSize{}{E_k\left(U_{k+1}^*\otimes\cdots\otimes U_m^*\otimes U_1^*\otimes \cdots \otimes U_{k-1}^*\right)}&\leq C(\sqrt{d_k}+\sqrt{r_{-k}}),\\
     \FnormSize{}{E_k\left(U_{k+1}^*\otimes\cdots\otimes U_m^*\otimes U_1^*\otimes \cdots \otimes U_{k-1}^*\right)}&\leq C\sqrt{d_kr_{-k}},\\
      \sup_{\substack{U_\ell\in\mathbb{O}_{d_\ell,r_\ell}\\\ell\neq [m]}}\SnormSize{}{E_k\left(U_{k+1}\otimes\cdots\otimes U_m\otimes U_1\otimes \cdots \otimes U_{k-1}\right)}&\leq C(\sqrt{d_k}+\sqrt{r_{-k}}+\sum_{\ell\neq k}\sqrt{d_\ell r_\ell}),\\
       \sup_{\substack{U_\ell\in\mathbb{O}_{d_\ell,r_\ell}\\\ell\neq [m]}}\FnormSize{}{E_k\left(U_{k+1}\otimes\cdots\otimes U_m\otimes U_1\otimes \cdots \otimes U_{k-1}\right)}&\leq C(\sqrt{d_k r_{-k}}+\sum_{\ell\neq k}\sqrt{d_\ell r_\ell}),\\
       \sup_{\substack{U_\ell\in\mathbb{O}_{d_\ell,r_\ell}\\\ell\neq [m]}}\FnormSize{}{\tE\times_1 U_1^T\times\cdots\times_m U_m^T}&\leq C(\sqrt{r_*}+\sum_{\ell\in [m]}\sqrt{d_\ell r_\ell}).
\end{align}
\end{lem}

\begin{lem}[Projection bound of perturbation]\label{lem:projec}
Suppose $X, E \in \mathbb{R}^{m\times n}$ and $\text{rank}(X) = r$
. Let $U \in\mathbb{O}_{m,r}$ be the leading r singular vectors of
$Y = X + E$. Then,
\begin{align}
    \SnormSize{}{(I-UU^T)X}&\leq 2\SnormSize{}{E},\\
    \FnormSize{}{(I-UU^T)X}&\leq \min\left(
2\sqrt{r}\SnormSize{}{E},2\FnormSize{}{E}\right).
\end{align}
\end{lem}
\begin{proof}[Proof of Lemma~\ref{lem:projec}]
For matrix norm bound, we have
\begin{align}\label{eq:snorm1}
 \SnormSize{}{(I-UU^T)X}&\leq \SnormSize{}{(I-UU^T)Y}+\SnormSize{}{E}\nonumber\\&\leq \sigma_{r+1}(Y)+\SnormSize{}{E}\nonumber\\&\leq \min_{Z\in\mathbb{R}^{m
 \times n}\colon \text{rank}(Z)\leq r}\SnormSize{}{Y-Z}+\SnormSize{}{E}\nonumber\\
 &\leq \SnormSize{}{Y-X}+\SnormSize{}{E}\nonumber\\&\leq 2\SnormSize{}{E}.
 \end{align}
 Similarly we bound the Frobenius norm
 \begin{align}
     \FnormSize{}{(I-UU^T)X}&\leq \FnormSize{}{(I-UU^T)Y}+\FnormSize{}{E}\\&\leq \sqrt{\sum_{i= r+1}^{m\wedge n}\sigma_{i}^2(Y)}+\FnormSize{}{E}\\&
     \leq \min_{Z\in\mathbb{R}^{m
 \times n}\colon \text{rank}(Z)\leq r}\FnormSize{}{Y-Z}+\FnormSize{}{E}\\
 &\leq \FnormSize{}{Y-X}+\FnormSize{}{E}\\&\leq 2\FnormSize{}{E}.
 \end{align}
 In addition, a direct application of \eqref{eq:snorm1} yields
 \begin{align}
     \FnormSize{}{(I-UU^T)X}\leq 2\sqrt{r}\SnormSize{}{E}.
 \end{align}
\end{proof}

\begin{lem}[Lemma 2 in \cite{han2022exact}]\label{lem:new}
Suppose the first $r$ and the rest $d_1 - r$ singular vectors of $Y \in \mathbb{R	}^{d_1\times d_2}$ are $\tilde U\in\mathbb{O}_{d_1,r}$ and $\tilde U_{\perp}\in\mathbb{O}_{d_1,d_1-r}$, respectively. For some $1\leq r'\leq r$, let $W\in\mathbb{O}_{d_1,r'}$ be any orthonomal matrix and $W_{\perp}\in\mathbb{O}_{d_1,r'}$ be the orthogonal complement of $W$. Given that 
 $\sigma_{r'}(W^TY)>\sigma_{r+1}(Y)$,  we have 
\begin{align}
     \SnormSize{}{\tilde U_{r\perp}^TW} \leq\frac{\sigma_{r'}(W^TY)\SnormSize{}{W_{\perp}^TY\mathbb{P}_{Y^TW}}}{\sigma_{r'}^2(W^TY)-\sigma_{r+1}^2(Y)}.
\end{align}
\end{lem}

\begin{lem}[Perturbation Bound on Subspaces of Different Dimensions]\label{lem:uu} Consider the signal plus noise model,  \[Y = X + X_{\perp}+ E\in\mathbb{R}^{d_1\times d_2},\]
where $X$ is a signal matrix such that $\text{rank}(X)=r$, $X_{\perp}$ is a perturbation matrix, and $E$ is a noise matrix with i.i.d. standard sub-Gaussian entries. 
Define \begin{align}\label{eq:rcond}r':=\max\{r'\in\{0,1,\ldots,r\}\colon \sigma_{r'}(X)\geq \max(\sqrt{3}(d_1+\sqrt{d_1d_2}),16\SnormSize{}{X_{\perp}})\}.\end{align}
We denote 
\begin{align}
        \tilde U_r =\text{SVD}_r(Y),\quad U_{r'}= \text{SVD}_{r'}(X).
\end{align} Then with probability at
least $1 -\exp(-cd_1 \wedge d_2)$,
\begin{align}
    \SnormSize{}{\tilde U_{r\perp}^TU_{r'}}\leq C\left(\frac{\sqrt{d_1}+\SnormSize{}{X_{\perp}}}{\sigma_{r'}(X)}+\frac{\sqrt{d_1 d_2}+\sqrt{d_1\vee d_2}\SnormSize{}{X_{\perp}}+\SnormSize{}{X_{\perp}}^2}{\sigma_{r'}^2(X)}\right),
\end{align}
where $\tilde U_{r\perp}\in\mathbb{R}^{d_1\times d_1-r}$ is the orthogonal complement matrix of $\tilde U_{r}$.
\end{lem}
\begin{proof}[Proof of Lemma~\ref{lem:uu}]
Applying Lemma~\ref{lem:new} with $W = U_{r'}$, we have
\begin{align}\label{eq:goal}
     \SnormSize{}{\tilde U_{r\perp}^TU_{r'}} \leq\frac{\sigma_{r'}(U_{r'}^TY)\SnormSize{}{U_{r'\perp}^TY\mathbb{P}_{Y^TU_{r'}}}}{\sigma_{r'}^2(U_{r'}^TY)-\sigma_{r+1}^2(Y)},
\end{align}
Therefore, it suffices  to provide the probabilistic bounds of  $\sigma_{r'}^2(U_{r'}Y)-\sigma_{r+1}^2(Y)$,  $\SnormSize{}{U_{r'\perp}^TY\mathbb{P}_{Y^TU_{r'}}}$, and $\sigma_{r'}(U_{r'}^TY)$.

First we denote
\begin{align*}
&X_1 = \mathbb{P}_{U_{r'}}X,\quad X_2 = X-X_1,\\
&Y = X+ X_{\perp}+ E,\quad Y' = X + E,\\&Y_1 = X_1+X_{\perp}+E,\quad Y_1' = X_1+E.
\end{align*}

We now provide the bounds of $\sigma_{r'}^2(U_{r'}Y)-\sigma_{r+1}^2(Y)$ and $\SnormSize{}{U_{r'\perp}^TY\mathbb{P}_{Y^TU_{r'}}}$, and $\sigma_{r'}(U_{r'}^TY)$.

  \paragraph*{Bound of the term $\sigma_{r'}^2(U_{r'}Y)-\sigma_{r+1}^2(Y)$} By \cite[Lemma 4]{cai2018rate}, for all $x>0$, we have 
    \begin{align}
        &\mathbb{P}\left(\sigma^2_{r'}(U_{r'}^TY_1')\leq(\sigma_{r'}^2(X_1)+d_2)(1-x)\right) \leq C\exp(Cr-c(\sigma_{r'}(X_1)+d_2)x^2\wedge x),\\
        &\mathbb{P}\left(\sigma^2_{r+1}(Y')\geq d_2(1+x)\right) \leq C\exp(Cd_1-cd_2x^2\wedge x),\\
        &\mathbb{P}\left(\SnormSize{}{U_{r'\perp}^TY_1\mathbb{P}_{Y_1^TU_{r'}}}\geq x\right)\leq C\exp\left(C\exp(Cd_1-cx^2\wedge x\sqrt{\sigma_{r'}^2(X+X_{\perp})+d_2}\right).
    \end{align}
        By setting $x$ as $\frac{\sigma_{r'}^2(X)}{3(\sigma^2_{r'}(X)+d_2)}$, $\frac{\sigma_{r'}^2(X)}{3d_2}$, and $C\left(\sqrt{d_1}+\frac{d_1}{\sqrt{\sigma_{r'}^2(X+X_{\perp})+d_2}}\right)$ respectively, we obtain
    \begin{align}\label{eq:uyp}
& \sigma_{r'}(U_{r'}^TY_1')\geq \sqrt{\frac{2\sigma_{r'}^2(X)}{3}+d_2},\quad \sigma_{r+1}(Y')\leq \sqrt{\frac{\sigma_{r'}^2(X)}{3}+d_2},\nonumber\\&\text{ and }  \SnormSize{}{U_{r'\perp}^TY_1\mathbb{P}_{Y_1^TU_{r'}}}\leq C\left(\sqrt{d_1}+\frac{d_1}{\sqrt{\sigma_{r'}^2(X+X_{\perp})+d_2}}\right),
    \end{align}
    with probability at least 1-$C\exp(-cd_1\wedge d_2)$.
    Since $Y = Y'+ X_{\perp}$ and $Y_1 = Y'_1+X_{\perp}$,  applying Weyl's inequality yields
    \begin{align}\label{eq:sigma}
        &\sigma_{r'}(U_{r'}^TY) = \sigma_{r'}(U_{r'}^TY_1)\geq \sqrt{\frac{2\sigma_{r'}^2(X)}{3}+d_2}-\SnormSize{}{X_{\perp}},\\&\sigma_{r+1}(Y)\leq \sqrt{\frac{\sigma_{r'}^2(X)}{3}+d_2}+\SnormSize{}{X_{\perp}}.
    \end{align}
    Therefore, we obtain the following inequality from \eqref{eq:sigma},
    \begin{align}\label{eq:s1}
         \sigma_{r'}^2(U_{r'}^TY)-\sigma^2_{r+1}(Y)&\geq \frac{\sigma_{r'}^2(X)}{3}-4\SnormSize{}{X_{\perp}}\sqrt{\frac{2\sigma_{r'}^2(X)}{3}+d_2}\nonumber\\
         &\geq \frac{\sigma_{r'}^2(X)}{3}-\frac{1}{4}\sigma_{r'}(X)\sqrt{\frac{2\sigma_{r'}^2(X)}{3}+d_2}\nonumber\\
         &\geq \frac{\sigma_{r'}^2(X)}{12},
    \end{align}
    where the second inequality uses  $\sigma_{r'}(X)\geq 16\SnormSize{}{X_{\perp}}$ while the last inequality uses  $\sigma_{r'}(X)\geq \sqrt{3d_2}$ by the definition of  $r'$ in  \eqref{eq:rcond}.    
   
  \paragraph*{Bound of $\SnormSize{}{U_{r'\perp}^TY\mathbb{P}_{Y^TU_{r'}}}$} 
  Since $\sigma_{r'}(X+X_{\perp})\geq \sigma_{r'}(X) - \SnormSize{}{X_\perp}\geq \frac{15}{16} \sigma_{r'}(X)$, we have    
\begin{align}
 \SnormSize{}{U_{r'\perp}^TY_1\mathbb{P}_{Y_1^TU_{r'}}}\leq C\left(\sqrt{d_1}+\frac{d_1}{\sqrt{\sigma_{r'}^2(X+X_{\perp})+d_2}}\right)\leq C\sqrt{d_1},
\end{align}
where we use the definition of  $r'$ and $C$ absorbs all constant factors.  
Notice that 
\begin{align}
\SnormSize{}{U_{r'\perp}^TY\mathbb{P}_{Y^TU_{r'}}}&= \SnormSize{}{U_{r'\perp}^TY\mathbb{P}_{Y_1^TU_{r'}}}\\&=\SnormSize{}{U_{r'\perp}^T(Y_1+X_2)\mathbb{P}_{Y_1^TU_{r'}}}\\&\leq \SnormSize{}{U_{r'\perp}^TY_1\mathbb{P}_{Y_1^TU_{r'}}}+\SnormSize{}{U_{r'\perp}^TX_2\mathbb{P}_{Y_1^TU_{r'}}}\\&\leq C\sqrt{d_1}+\SnormSize{}{U_{r'\perp}^TX_2\mathbb{P}_{Y_1^TU_{r'}}}.
\end{align}
Now we focus on bounding $\SnormSize{}{U_{r'\perp}^TX_2\mathbb{P}_{Y_1^TU_{r'}}}$. We have
\begin{align}\label{eq:uxp}
\SnormSize{}{U_{r'\perp}^TX_2\mathbb{P}_{Y_1^TU_{r'}}}&\leq \frac{\SnormSize{}{U_{r'\perp}^TX_2(Y_1^TU_{r'})}}{\sigma_{r'}(Y_1^TU_{r'})}\nonumber\\&\leq \frac{\SnormSize{}{U_{r'\perp}^TX_2X_{\perp}^TU_{r'}}+\SnormSize{}{U_{r'\perp}^TX_2E^TU_{r'}}}{\sigma_{r'}(Y_1^TU_{r'})}\nonumber\\&\leq C \frac{\SnormSize{}{U_{r'\perp}^TX_2X_{\perp}^TU_{r'}}+\SnormSize{}{U_{r'\perp}^TX_2E^TU_{r'}}}{\sigma_{r'}(X)},
\end{align}
where we use $\sigma_{r'}(Y_1^TU_{r'})\geq \frac{15}{16}\sigma_{r'}(X)$ from combining \eqref{eq:sigma} and definition of $r'$ in  \eqref{eq:rcond}. 
 In addition,  by Lemma~\ref{lem:subg}, we have the following with probability $1-C\exp(-cd_1)$, 
 \begin{align*}
 \SnormSize{}{U_{r'\perp}^TX_2E^TU_{r'}}\leq C\SnormSize{}{X_2}\sqrt{d_1}.
 \end{align*}
Thus, we have the bound  of  $\SnormSize{}{U_{r'\perp}^TX_2\mathbb{P}_{Y_1^TU_{r'}}}$ from \eqref{eq:uxp} as 
 \begin{align}
 \SnormSize{}{U_{r'\perp}^TX_2\mathbb{P}_{Y_1^TU_{r'}}}&\leq C\frac{\SnormSize{}{X_2}(\
\SnormSize{}{X_\perp}+\sqrt{d_1})}{\sigma_{r'}(X)}\leq C(\
\SnormSize{}{X_\perp}+\sqrt{d_1}),
 \end{align}
 where the last inequality uses $\SnormSize{}{X_2}\leq \sigma_{r'}(X)$.
 
 Finally, we obtain the bound of $\SnormSize{}{U_{r'\perp}^TY\mathbb{P}_{Y^TU_{r'}}}$ as 
 \begin{align}\label{eq:s2}
 \SnormSize{}{U_{r'\perp}^TY\mathbb{P}_{Y^TU_{r'}}}\leq C(\SnormSize{}{X_\perp}+\sqrt{d_1})
 \end{align}
  
\paragraph*{Bound of $\sigma_{r'}(U_{r'}^TY)$} 
We  obtain the upper bound by Weyl's inequality,
    \begin{align}\label{eq:s3}
  \sigma_{r'}(U_{r'}^TY)&\leq \sigma_{r'}(X)+\SnormSize{}{U_{r'}^TX_{\perp}}+\SnormSize{}{U_{r'}^TE}\\&\leq \sigma_{r'}(X)+\SnormSize{}{X_{\perp}}+\sqrt{d_2},
    \end{align}
    where the last inequality holds with probability at least $1-C\exp(-cd_2)$   by Lemma~\ref{lem:subg}.
    
Finally, plugging inequalities \eqref{eq:s1},\eqref{eq:s2}, and \eqref{eq:s3} into \eqref{eq:goal} yeilds,
\begin{align}
     \SnormSize{}{\tilde U_{r\perp}^TU_{r'}}\leq C\left(\frac{\sqrt{d_1}+\SnormSize{}{X_{\perp}}}{\sigma_{r'}(X)}+\frac{\sqrt{d_1 d_2}+\sqrt{d_1\vee d_2}\SnormSize{}{X_{\perp}}+\SnormSize{}{X_{\perp}}^2}{\sigma_{r'}^2(X)}\right).
\end{align}
\end{proof}

\section{Additional explanation of crop production analysis}\label{sec:cropappd}
We perform clustering analyses based on the Tucker representation of the estimated signal tensor $\hat\Theta$. The procedure is motivated from the higher-order extension of Principal Component Analysis (PCA). Recall that, in the matrix case, we perform clustering on a matrix $\mX\in\mathbb{R}^{m\times n}$ based on the following procedure. First, we factorize $X$ into
\begin{equation}
\mX = \mU\mSigma \mV^T,
\end{equation}
where $\mSigma$ is a diagonal matrix and $\mU,\mV$ are factor matrices with orthogonal columns. Second, we take each column of $\mV$ as a principal axis and each row in $\mU\mSigma$ as principal component. A subsequent multivariate clustering method (such as $K$-means) is then applied to the $m$ rows of $\mU\mSigma$.

We apply a similar clustering procedure to the estimated signal tensor $\hat\Theta$. We factorize $\hat \Theta$ based on Tucker decomposition.
\begin{equation}\label{eq:Tuckerest}
\hat \Theta = \hat \tC\times_1\hat \mU_1\times_2\cdots\times_m\hat \mU_m,
\end{equation}
where $\hat \tC\in\mathbb{R}^{r_1\times \cdots \times r_m}$ is the estimated core tensor, $\hat \mU_k\in\mathbb{R}^{d_k\times r_k}$ are estimated factor matrices with orthogonal columns.  The mode-$k$ unfolding of~\eqref{eq:Tuckerest} gives
\begin{equation}
\textup{Unfold}_k(\hat \Theta) = \hat \mU_k\textup{Unfold}_k(\hat \tC)\left(\hat \mU_m\otimes\cdots\otimes\hat \mU_1\right).
\end{equation}
 We conduct clustering on this  mode-$k$ unfolded signal tensor. We take columns in $\left(\hat \mU_m\otimes\cdots\otimes\hat \mU_1\right)$ as principal axes and rows in $\hat \mU_k\textup{Unfold}_k(\hat \tC)$ as principal components. Finally, we apply $K$-means clustering method to the $d_k$ rows of the matrix $\hat \mU_k\textup{Unfold}_k(\hat \tC)$.
We pick the number of clusters based on the elbow method. Figure~\ref{fig:elbow}  suggests six and five clusters on country and crop respectively.
\begin{figure}[H]
     \centering
     \begin{subfigure}[b]{0.45\textwidth}
         \centering
         \includegraphics[width=\textwidth]{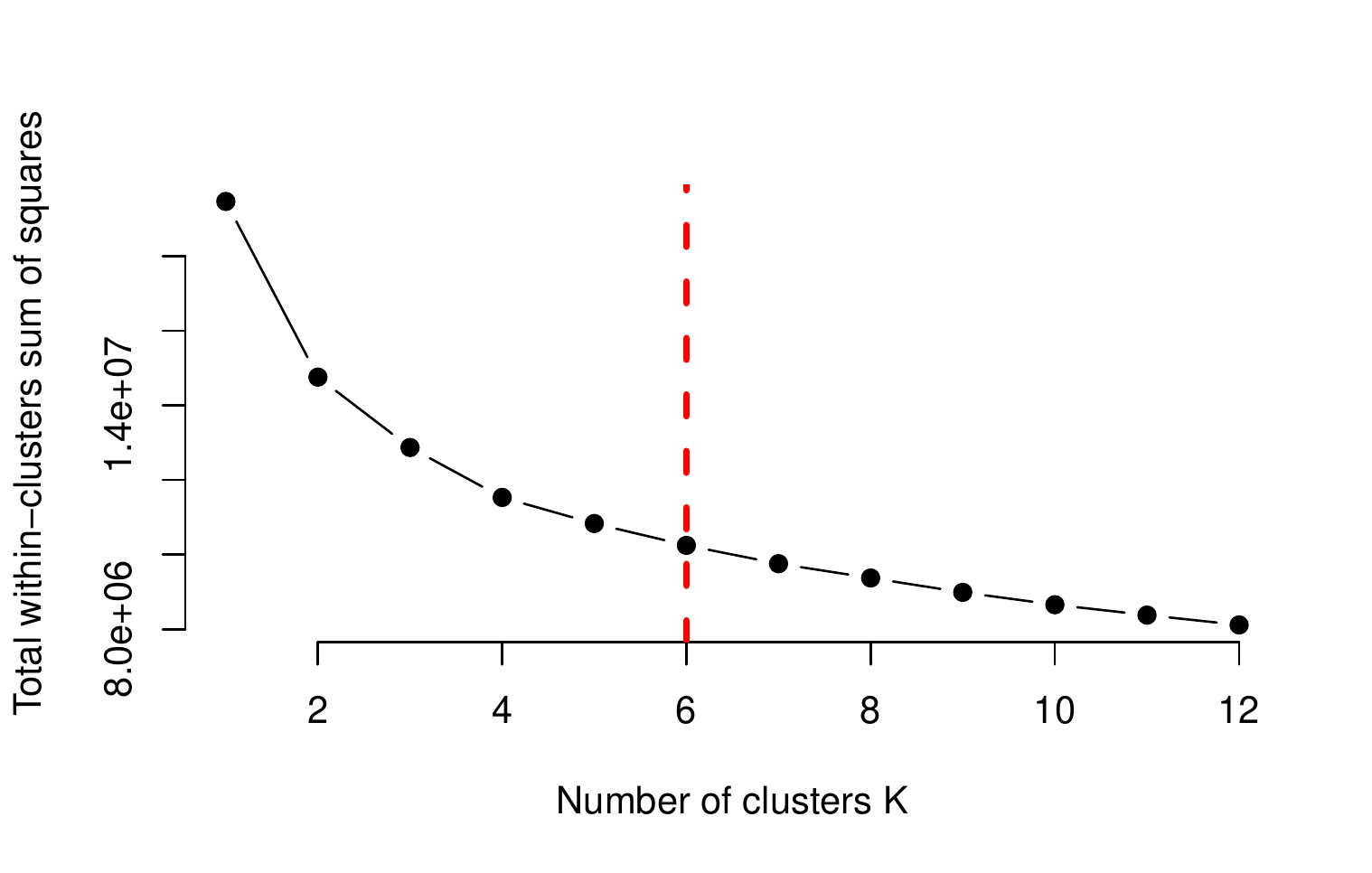}
         \caption{The elbow plot for country clustering}
     \end{subfigure}
     \hfill
     \begin{subfigure}[b]{0.45\textwidth}
         \centering
         \includegraphics[width=\textwidth]{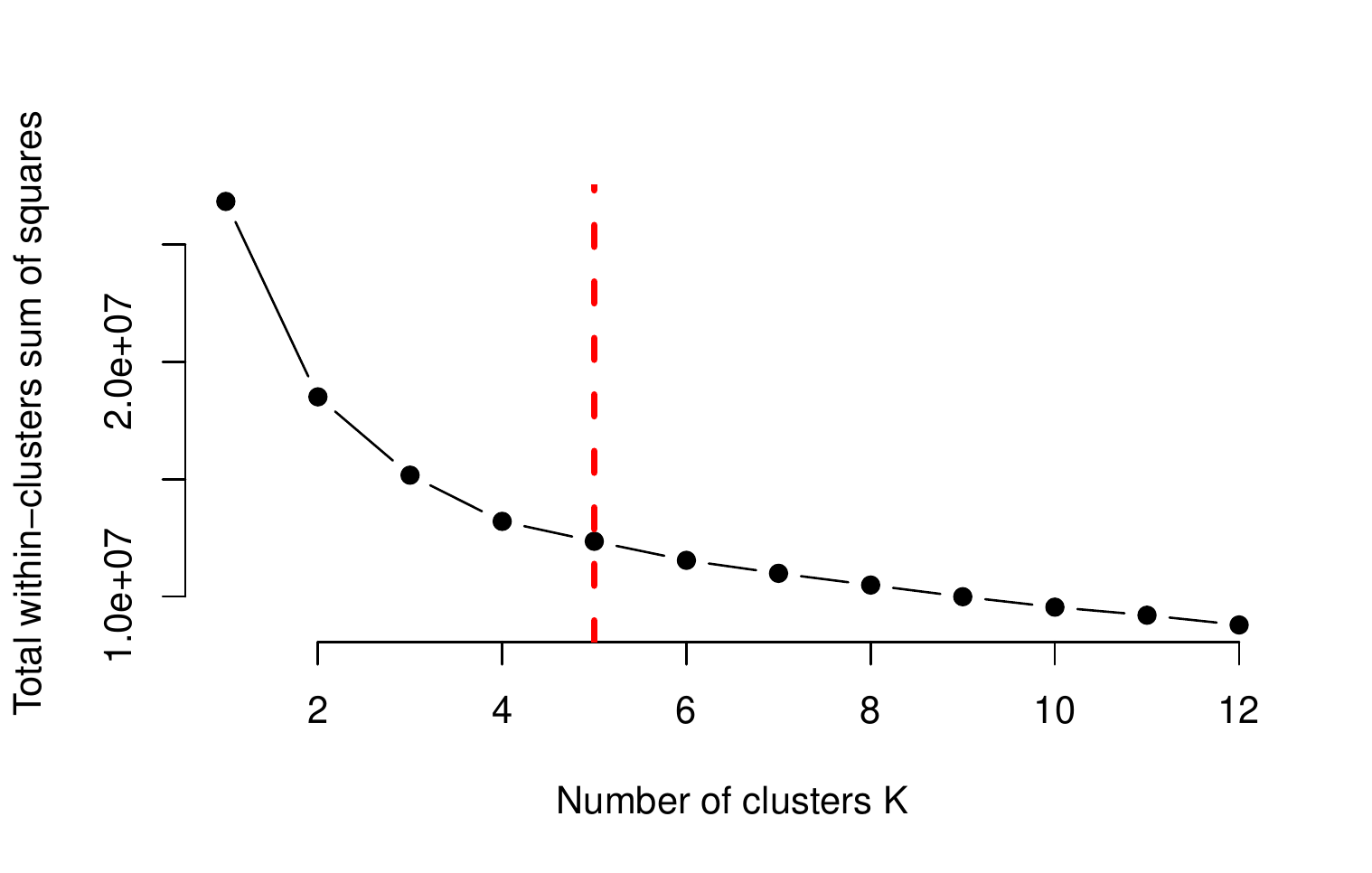}
         \caption{The elbow plot for crop clustering}
         
     \end{subfigure}
        \caption{Elbow plots for determining the number of clusters in $K$-means.}
        \label{fig:elbow}
\end{figure}
Clustering on countries are investigated in the main paper. Here, we provide the clustering results on crops.  Table~\ref{tab:croptb} summarizes the five clusters of crop items. We find that five clusters  captures the similar type of crops. For example, Cluster 3 represents berries and leafy plants whereas Cluster 4 consists of crops mainly produced in Asia region.
 
\begin{table}[H]
    \centering
    \caption{Five clusters of 161 crops  based on the estimated signal tensors in crop production data application}
    \includegraphics[width = \textwidth]{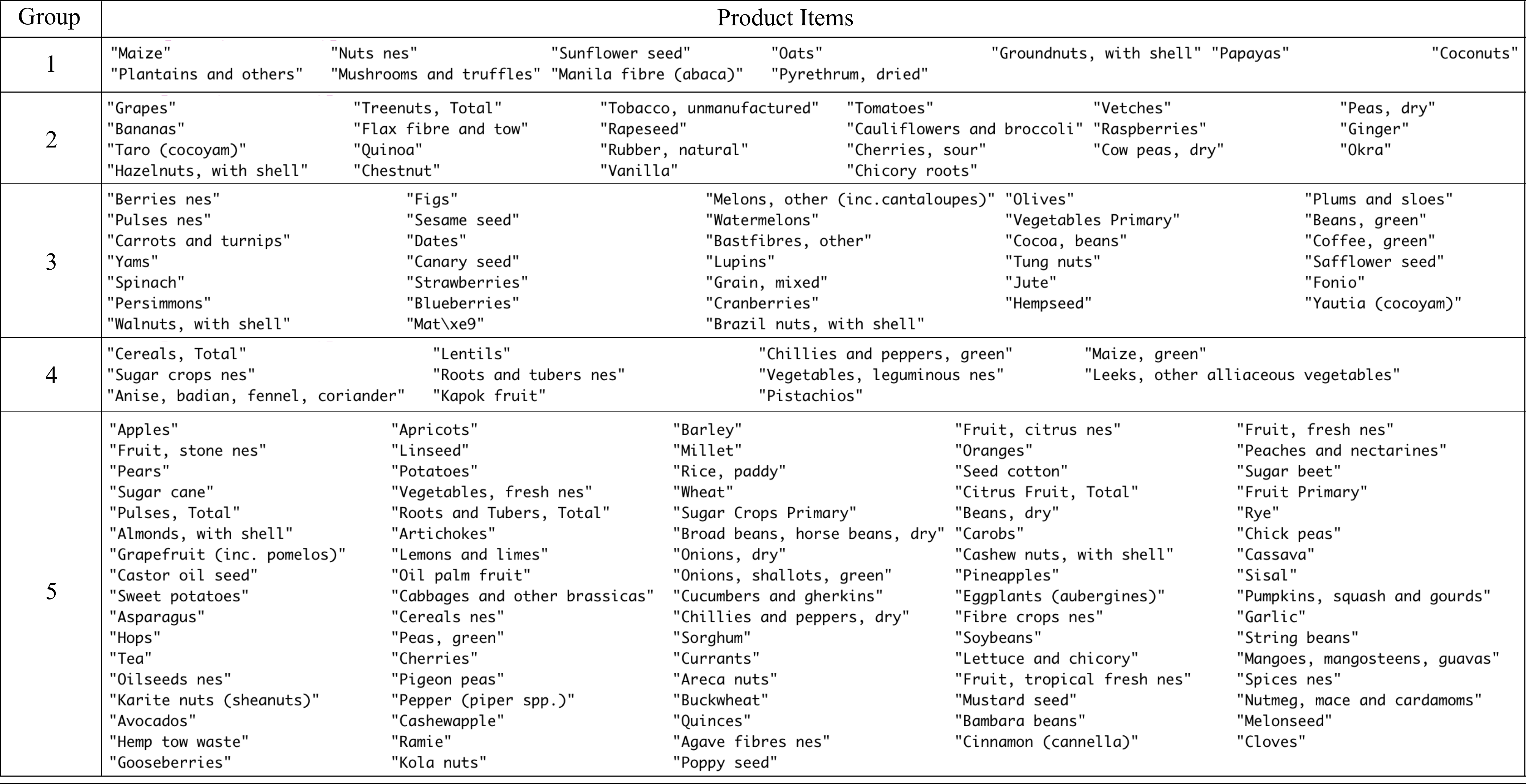}
    \label{tab:croptb}
\end{table}

\end{document}